\documentclass[11pt]{article}

\usepackage[ngerman, english]{babel}
\usepackage{amsmath,amssymb,amsthm,stmaryrd,cancel,dsfont,scalerel}
\usepackage{bbm}
\usepackage{mathdots}

\RequirePackage[round]{natbib}

\usepackage{tikz}
\usetikzlibrary{matrix,arrows,shapes,decorations.shapes}
\usetikzlibrary{decorations.pathmorphing}

\usetikzlibrary{spy}
\usetikzlibrary{backgrounds}
\usetikzlibrary{decorations}
\usetikzlibrary{patterns}

\pgfdeclarelayer{background layer}
\pgfsetlayers{background layer,main}
\usepackage{xcolor,paralist,graphicx}

\usepackage{geometry}
\usepackage{stmaryrd}
\usepackage[latin1]{inputenc}
\usepackage{paralist}
\usepackage{array}
\usepackage{multirow,booktabs}
\usepackage{relsize}
\usepackage{hyperref, aliascnt}
\usepackage{subcaption}
\usepackage{algorithm}
\usepackage{algorithmic}
\usepackage{eqparbox}
\usepackage{float}

\newcommand{\E}{{\mathbb E}}
\newcommand{\F}{{\widetilde F}}
\newcommand{\N}{{\mathbb N}}
\newcommand{\Q}{{\mathbb Q}}
\newcommand{\R}{{\mathbb R}}

\newcommand{\G}{{\mathcal G}}
\newcommand{\T}{{\widetilde T}}

\newcommand{\xit}{{\widetilde \xi}}
\newcommand{\Lt}{{\widetilde L}}
\renewcommand{\O}{{\mathcal O}}
\renewcommand{\S}{{\mathbb S}}

\theoremstyle{definition}
\newtheorem{defn}{Definition}[section]
\newtheorem{rem}[defn]{Remark}

\newtheorem{spc}[defn]{Specification}

\newtheorem{prop}[defn]{Proposition}
\newtheorem{thm}[defn]{Theorem}

\newcommand{\sref}[2]{\hyperref[#2]{#1 \ref*{#2}}}

\hypersetup{
    colorlinks=true,
    linkcolor={blue!80!black},
    filecolor=cyan,      
    urlcolor=magenta,
    citecolor={green!70!black},
}

\usepackage{enumitem}
\setlist[enumerate]{leftmargin=.5in}
\setlist[itemize]{leftmargin=.5in}

\usepackage{amsopn}
\DeclareMathOperator{\diag}{diag}

\DeclareMathOperator*{\Bigcdot}{\scalerel*{\cdot}{\bigodot}}

\newcommand\scalemath[2]{\scalebox{#1}{\mbox{\ensuremath{\displaystyle #2}}}}

\makeatletter
\renewcommand*\env@matrix[1][\arraystretch]{%
  \edef\arraystretch{#1}%
  \hskip -\arraycolsep
  \let\@ifnextchar\new@ifnextchar
  \array{*\c@MaxMatrixCols c}}
\makeatother

\title{A multi-factor polynomial framework for long-term electricity forwards with delivery period\thanks{This research is part of a collaboration with Axpo Solutions AG. Xi Kleisinger-Yu gratefully acknowledges the financial support and electricity data provided by Axpo Solutions AG.}}

\author{Xi Kleisinger-Yu\thanks{ETH Zurich, Department of Mathematics, R\"amistrasse 101, 8092 Zurich, Switzerland, {xi.kleisinger-yu@math.ethz.ch}.}
  \and Vlatka Komaric\thanks{Axpo Solutions AG, Risk Management and Valuation Department, Parkstrasse 23, 5400 Baden, Switzerland,
  {vlatka.komaric@axpo.com}.} 
\and Martin Larsson\thanks{Carnegie Mellon University, Department of Mathematical Sciences, Pittsburgh, Pennsylvania 15213, USA, {martinl@andrew.cmu.edu}.}
  \and Markus Regez\thanks{Axpo Solutions AG, Risk Management and Valuation Department, Parkstrasse 23, 5400 Baden, Switzerland, {markus.regez@axpo.com}.}}
\date{June 09, 2020 \medskip \\ forthcoming in SIAM Journal on Financial Mathematics}

\begin{document}

\maketitle


\begin{abstract}
We propose a multi-factor polynomial framework to model and hedge long-term electricity contracts with delivery period. This framework has several advantages: the computation of forwards, risk premium and correlation between different forwards are fully explicit, and the model can be calibrated to observed electricity forward curves easily and well. Electricity markets suffer from non-storability and poor medium- to long-term liquidity. Therefore, we suggest a rolling hedge which only uses liquid forward contracts and is risk-minimizing in the sense of F\"ollmer and Schweizer. We calibrate the model to over eight years of German power calendar year forward curves and investigate the quality of the risk-minimizing hedge over various time horizons.
\end{abstract}



\section{Introduction}

Electricity differs from other energy commodities due to specific features such as limited storability, possibility of intra-day and day-ahead negative prices, its unique mechanism of the auction market, high liquidity of short- to medium-term trading and illiquidity of its long-term trading. Much of the academic literature is dedicated to short- to medium-term modeling of electricity spot and futures prices, as its highly frequent and huge data amount makes it ideal for empirical studies of time series analysis. However, the literature addressing the modeling of long-term electricity forwards and the corresponding hedging problems is scarce.

In this paper, we propose a mathematically tractable multi-factor polynomial diffusion framework to model long-term forwards, which captures long-term properties such as mean reversion well. In this framework the computation of forwards and cross-maturity correlations are fully explicit. Fitting the model to long time series of single market electricity data works easily and well. Furthermore, we set up a rolling hedge mechanism that only uses liquid forward contracts. This allows us to address the non-storability of electricity and poor liquidity in its long-term markets. Within the setup the hedging strategy we suggest minimizes the conditional variance of the cost processes at any time, and thus is risk-minimizing in the sense of F\"ollmer and Schweizer. A simulation study using the estimated model shows that the risk-minimizing rolling hedge significantly reduces, yet does not fully eliminate, the variance and skew of the long-term exposures.  

The proposed modeling framework has various applications in forward modeling. It can be used to smoothly extrapolate the curve to the non-liquid horizon while calibrating it to the liquid horizon; it can also be used to smooth the forward surface implied by the market once calibrated and to filter out market noise; moreover, it can be used to model the prices within the real data horizon between two quotation dates. Furthermore, the model can be extended to model multiple electricity markets and other energy markets simultaneously. It can thus serve as an alternative model for risk management purposes, and for conducting simulations. We do however not pursue such multi-market extensions in this paper.
 
Compared to other electricity modeling classes such as affine processes (mostly used as geometric models), this modeling framework has the advantage of being general but still very tractable, so that pricing formulas of spots, forwards (with instantaneous delivery) and forwards with delivery period have closed-form solutions. Moreover, it is possible to explicitly compute locally risk-minimizing hedging strategies in this framework which uses a rolling mechanism.

Our framework is introduced to model long-term markets and yearly forward contracts, which are the most liquidly traded long-term contracts. The primary focus is to capture dynamics over very long time horizons, including contracts with maturities far beyond the liquidity of long-term futures traded on the exchange. We calibrate the model to over-the-counter forwards with maturities of up to ten years from the quotation date. However, our framework can easily be extended to capture features such as spikes, seasonality and negative prices for spots and forwards with shorter time-to-maturity (day-ahead, week-ahead, month-ahead, quarter-ahead) and with shorter time frames (daily, monthly, quarterly). Incorporating such features does not change the polynomial structure, so that pricing and hedging remains tractable.

Polynomial models have been used to solve a number of problems in finance, see \cite{filipovic2017linear,ackerer2017option,ackerer2018jacobi,cuchiero2018polynomial,filipovic2018term,ackerer2016linear,filipovic2016quadratic,biagini2016polynomial,delbaen2002interest} for references as well as \cite{cuchiero2012polynomial,Larsson2016} for a treatment of the underlying mathematical theory. With the exception of \cite{filipovic2018polynomial}, the polynomial processes have not been used for electricity modeling. 
Our polynomial framework makes assumption on properties of spot and forward and not on supply-demand relation, and thus falls into the category of classical reduced-form model (see \cite{carmona2014survey} for details on reduced-form model versus structural approach). It is closest to the arithmetic models of  \cite{benth2007extracting,benth2007non,benth2008stochastic}, and extends them by making the spot price not a linear combination but a squared combination of underlying polynomial processes. In doing so we extend the class of stochastic process on the one hand, and guarantee non-negative spot prices on the other hand.  

The local risk-minimization hedging criterion of F\"ollmer \& Schweizer 1991 is one the two main quadratic hedging approaches; see e.g.\ \cite{follmer1991hedging,heath1999,schweizer1999guided,schweizer1990risk} for references of the general theory of local risk-minimization, \cite{follmer1986contributions} for the mean-variance hedge, and \cite{heath2001comparison} for a comparison of the two approaches. In a recent paper on hedging, a locally risk-minimizing hedge was given for the arithmetic model of Benth et al.\ under illiquidity; see \cite{christodoulou2018local}. Our work differs from theirs, as we consider a rolling hedge which only uses liquid forward contracts and give explicit expression for the locally risk-minimizing hedging strategy for our modeling framework.

This paper is structured in the following way: In Section 2, we define the underlying polynomial framework, model the spot price as a quadratic function of it, and provide two main specifications. In Section 3, we briefly review the main characteristics of polynomial diffusions, with a focus on the moment formula for polynomials of degree two. In Section 4, we define electricity forwards with and without delivery period. We give pricing formulas for forwards, as well as explicit expressions for covariances and correlations between different forwards. In Section 5, to incorporate time series observations of forward prices, we specify a market price of risk function, which determines the forward price dynamics under the real-world measure $\mathbb{P}$, and define the forward risk premium. In Section 6, we introduce a rolling hedge mechanism with liquidity constraints for hedging a long-term electricity commitment. Further, we give a rolling-hedge that is locally risk-minimizing in the sense of F\"ollmer and Schweizer. In Section 7, we perform model estimation of a specification of the polynomial framework to a time series of real observations of power forwards using a quadratic Kalman filter. Further we simulate forward curves and investigate the quality of the risk-minimizing hedge over various time horizon.      

Throughout this paper, we fix a filtered probability space $(\Omega, \mathcal{F}, \mathcal{F}_t, \mathbb{Q})$, where $\Q$ is a risk-neutral probability measure used for pricing. For simplicity we assume zero interest rate and thus apply no discounting. We denote by $\mathbb{S}^d$ the set of all symmetric $d\times d$ matrices and $\mathbb{S}^d_+$ the subset consisting of positive semidefinite matrices. We let $\text{Pol}_n$ denote the space of polynomials on $\R^d$ of degree at most $n$.


\section{The model}\label{sec_model}
In this section we define the underlying polynomial framework. Firstly, we model the spot price $S_t$ as a quadratic function of an underlying $d$-dimensional state variable $X_t$ which evolves according to a polynomial diffusion. More precisely, we let
\begin{align}
  S_t  &= p_S(X_t)  \label{St_def} \\
  dX_t &= \kappa (\theta - X_t) dt + \sigma(X_t) dW_t \label{X_SDE}
\end{align}
where $p_S(x)= c+x^\top Q x$ with $c\in \R_+$ and $Q\in \mathbb{S}^d_+$, $\kappa \in \R^{d \times d}$, $\theta \in \R^d$, $W$ a $d$-dimensional Brownian motion under $\Q$ and $\sigma: \R^d \rightarrow \R^{d\times d}$ is continuous. We assume that the components of the diffusion matrix $a(x) := \sigma(x)\sigma(x)^\top$ are polynomials of degree at most two. This ensures that $X_t$ is a polynomial diffusion, see Lemma 2.2 in \cite{Larsson2016}.

The above formulation allows in particular to capture mean reversion, an important feature of electricity price dynamics. Empirically, this has been backed up by e.g.\ \cite{koekebakker2005forward}. They examined Nordic electricity forwards from 1995--2001 and observed that the short-term price varies around the long-term price, indicating mean reversion. 
Several economic arguments also support the mean-reverting property; see e.g.\ \cite{escribano2011modelling}.

We will now focus on the following two specifications.

\begin{spc}[Two-factor model] \label{Sp_1}
Let $\kappa_Z$, $\kappa_Y\in\R$, $\sigma_Z,\sigma_Y>0$, and $\rho\in(-1,1)$. The process $X_t := (Z_t, Y_t)^\top$ evolves according to the SDE
\begin{equation}\label{Eq_2factor}
\begin{aligned}
dZ_t &= -\kappa_Z Z_t dt + \sigma_ZdW^{(1)}_t \\
dY_t &= \kappa_Y(Z_t - Y_t)dt + \rho\sigma_Y dW^{(1)}_t + \sigma_Y\sqrt{1-\rho^2} dW^{(2)}_t
\end{aligned}
\end{equation}
with $Z_0,Y_0\in\R$ and $W_t = ( W^{(1)}_t,  W^{(2)}_t)^\top$ a standard two-dimensional Brownian motion. Here $Y_t$ mean-reverts at rate $\kappa_Y$ towards the correlated process $Z_t$. And thus, $Y_t$ and $Z_t$ can be seen as factor processes that drive the short-end and long-end dynamics of spot prices respectively.  This model is consistent with the empirical findings by \cite{koekebakker2005forward} regarding mean reversion. 
Let $\alpha$, $\beta$, $c\in\R_{+}$ and let the spot price be given by
\begin{align*}
S_t := c + \alpha Y^2_t + \beta Z^2_t. 
\end{align*}
This guarantees nonnegative spot price, as $S_t \geq c\ge0$. This specification is of the form \eqref{St_def}--\eqref{X_SDE} with 
\begin{equation}\label{Eq_sigmax1}
Q = \begin{pmatrix}
\beta & 0 \\ 0 & \alpha
\end{pmatrix}, \
\kappa = \begin{pmatrix}
 \kappa_Z & 0 \\
-\kappa_Y & \kappa_Y  
\end{pmatrix}, \
\theta = \begin{pmatrix}
0\\ 0
\end{pmatrix}, \
 \sigma(x)=\sigma(z,y)= 
\begin{pmatrix}
\sigma_Z & 0\\
\rho\sigma_Y & \sigma_Y\sqrt{1 - \rho^2}  
\end{pmatrix}.
\end{equation}

\end{spc}

\begin{spc}[Three-factor model] \label{Sp_2}
We now present a specification which extends the two-factor model by modeling correlation between the underlying processes stochastically via a Jacobi process. Conditions under which the model exists and is unique are given below. Let $\kappa_Z$, $\kappa_Y\in\R$, $\kappa_R,\sigma_Z,\sigma_Y,\sigma_R>0$, and $\theta_R\in(-1,1)$. The process $X_t := (Z_t, Y_t, R_t)^\top$ evolves according to the SDE
\begin{equation}\label{Eq_spc2_SDE}
\begin{aligned}
dZ_t &= -\kappa_Z Z_t dt + \sigma_ZdW^{(1)}_t\\
dY_t &= \kappa_Y(Z_t - Y_t)dt + R_t\sigma_Y dW^{(1)}_t + \sigma_Y\sqrt{1-R_t^2} dW^{(2)}_t\\
dR_t &= \kappa_R (\theta_R - R_t)dt + \sigma_R \sqrt{1- R^2_t}dW^{(3)}_t
\end{aligned}
\end{equation}
with $Z_0,Y_0\in\R$, $R_0\in(-1,1)$, and $W_t =(W_t^{(1)}, W_t^{(2)}, W_t^{(3)})^\top$ a standard three-dimensional Brownian motion. Let $\alpha$, $\beta$, $c\in\R_{+}$ and let the spot price be given by
$$S_t := c + \alpha Y^2_t + \beta Z^2_t.$$
This specification is of the form \eqref{St_def}--\eqref{X_SDE} with
\begin{equation}\label{Eq_sigmax2}
\begin{aligned}
Q= \begin{pmatrix}  \beta & 0 & 0\\ 0 & \alpha & 0 \\ 0 & 0 & 0 \end{pmatrix},
\kappa = \begin{pmatrix}
\kappa_Z  & 0        &  0 \\
-\kappa_Y & \kappa_Y &  0 \\
0        &   0       &  \kappa_R 
\end{pmatrix} , \, 
\theta = \begin{pmatrix}
0 \\ 0 \\ \theta_R
\end{pmatrix}, \\
\sigma(x) = \sigma(z,y,r)= 
\begin{pmatrix}
\sigma_Z    & 0                        & 0 \\
r\sigma_Y & \sigma_Y\sqrt{1 - r^2} & 0 \\
0                         &  0           & \sigma_R\sqrt{1 - r^2}
\end{pmatrix}.
\end{aligned}
\end{equation}
\end{spc}

\begin{rem}
Although \sref{Specification}{Sp_2} is not used in our empirical analysis, we include it as an illustration of the flexibility of the polynomial framework.

A possible use of \sref{Specification}{Sp_2} is to model multi-energy commodities simultaneously. Here is a simple illustration of this: let one factor ($Z_t$) drive the short-term price of one market, and let the other factor ($Y_t$) drive the short-term price of the other market. Since energy markets evolve dynamically and prices are generally non-stationary over time (\cite{krevcar2019towards}), it is useful to have stochastic correlation between (the short ends of) different markets, modeled by a factor ($R_t$). The setup could be complemented with a fourth factor driving common long-term prices.

Alternatively, two markets can also be modeled as follows: two factors with the dynamics of $Y_t$, ($Y^i_t$, $i=1,2$), can be used to model short-term prices of each market; one factor ($Z_t$) drives the common long-end prices. In order to account for the changing relationship between short-term and long-term prices, another two factors with the dynamics of $R_t$, ($R^i_t$, $i=1,2$), can be added to model the stochastic correlation between the short-term and long-term prices in each market.

\end{rem}

\begin{prop}
Recall that $\kappa_R >0$, $\theta_R\in(-1,1)$, and assume moreover that
\begin{align}
\kappa_R(1+\theta_R)&\geq \sigma^2_R, \label{Rt_cond1}\\
\kappa_R(1-\theta_R)&\geq \sigma^2_R\label{Rt_cond2}.
\end{align}
Then for any initial condition with $Z_0\in\R$, $Y_0\in\R$ and $R_0\in(-1,1)$, there exists a unique strong solution $X_t = (Z_t, Y_t, R_t)^\top$ of the SDE \eqref{Eq_spc2_SDE}. Furthermore, this solution satisfies $R_t\in (-1,1)$ for all $t\ge0$.
\end{prop}

\begin{proof}
In the following we show the existence and uniqueness of a strong solution $R_t$ as well as its boundary non-attainment. Once this is shown, we can explicitly find $\widetilde{X}_t :=(Y_t, Z_t)$ in terms of $R_t$. Indeed, It\^o's formula yields
\[
d\left( e^{\tilde{\kappa} t}\widetilde{X}_t \right)= e^{\tilde{\kappa}t}\tilde{\kappa}\theta dt +  e^{\tilde{\kappa} t} \tilde{\sigma}(R_t) d\widetilde{W}_t,
\]
where $\tilde{\kappa} = \begin{pmatrix}
\kappa_Z  & 0        \\
-\kappa_Y & \kappa_Y  
\end{pmatrix}$
and $\tilde{\sigma}(r) = 
\begin{pmatrix}
\sigma_Z    & 0 \\
r\sigma_Y & \sigma_Y\sqrt{1 - r^2}
\end{pmatrix} $, which implies that
\[
\widetilde{X}_t = e^{-\tilde{\kappa} t} \widetilde{X}_0 + \int_0^t e^{-\tilde{\kappa} (t-s)}\tilde{\kappa} \theta ds + \int_0^t e^{-\tilde{\kappa} (t-s)}\tilde{\sigma}(R_t) d\widetilde{W}_s.
\]

We now prove existence of a weak solution of the SDE for $R_t$. Let $\varphi(r)$ be a continuous function that is equal to one for $r\in[-1,1]$ and is equal to zero for $|r|>2$, for example
\begin{align*}
\varphi(r)=
\begin{cases}
1 & |r|\leq 1\\
2-|r| & 1< |r|\leq2\\
0 & |r|>2.
\end{cases}
\end{align*}
We let  
$\tilde{b}(r):= b(r)\varphi(r)$ with $b(r):=\kappa_R(\theta_R - r)$ and $\tilde{\sigma}(r):=\sigma_R\sqrt{(1-r^2)_+}$. Then $\tilde{b}(r)$ and $\tilde{\sigma}(r)$ are continuous and bounded, and hence an $\R$-valued weak solution $R_t$ exists for the SDE $dR_t=\tilde{b}(R_t)dt+\tilde{\sigma}(R_t)dW^{(3)}_t$; see Theorem 4.22 of Section 5.4D in \cite{karatzas1998brownian}. We next show that $R_t$ stays in $(-1,1)$ using a version of ``McKean's argument". Let $p(r):=1-r^2$ and note that $p(R_0)>0$. Further define the stopping times $\tau_n:= \inf\{t: p(R_t)\leq \frac{1}{n}\}$ and $\tau:=\lim_{n\rightarrow \infty}\tau_n$. Observe that \eqref{Rt_cond1}--\eqref{Rt_cond2} imply that $\kappa_R(1-\theta_RR_t)-\sigma_R^2\ge0$ for all $t<\tau$. Combined with It\^o's formula, this yields
\begin{align*}
d\log p(R_{t})  &=\left( -(2\kappa_R-\sigma_R^2) + 2\frac{\kappa_R(1-\theta_RR_t)-\sigma_R^2}{1-R_t^2} \right) dt - \frac{2\sigma_R R_{t}}{\sqrt{1-R^2_{t}}}dW^{(3)}_{t}  \\
&\geq -(2\kappa_R-\sigma_R^2) dt- \frac{2\sigma_R R_{t}}{\sqrt{1-R^2_{t}}}dW^{(3)}_{t},
\end{align*} 
for $t< \tau$. Consider the process
\[
M_{t} := \int_0^{t}  \frac{2\sigma_R R_s}{\sqrt{1-R^2_s}}dW^{(3)}_s, \quad t<\tau.   
\]
Then $M_{t}$ is a local martingale on the stochastic interval $[0,\tau)$. By definition, this means that for all $n\in\N$, $M_{t\wedge\tau_n}$ is a local martingale. We now show that 
$\tau = \infty$ a.s. Suppose for contradiction that $\mathbb{P}(\tau<\infty)>0$. Then there exists a large $T<\infty$ such that $\mathbb{P}(\tau<T)>0$. Note that
\begin{equation}\label{eq_SDE_proof_12543}
M_t \ge -(2\kappa_R-\sigma_R^2)t+\log p(R_0) - \log p(R_t) \ge -(2\kappa_R-\sigma_R^2)T+\log p(R_0)
\end{equation}
for all $t< T\wedge\tau$. Thus $M_{t\wedge T}$ is uniformly bounded from below, and hence a local supermartingale on the stochastic interval $[0,\tau)$. The supermartingale convergence theorem for processes on stochastic interval $[0,\tau)$ now gives that $\lim_{t\rightarrow \tau}  M_{t\wedge T}$ exists in $\R$ almost surely; see e.g.\ the proof of Theorem~5.7 in \cite{Larsson2016}. Hence, in view of \eqref{eq_SDE_proof_12543}, $-\log p(R_t)$ is pathwise bounded above on $[0,T\wedge\tau)$, which in turn means that $\tau>T$ a.s. This contradiction shows that $R_t\in(-1,1)$ for all $t\geq0$.

Now let $\sigma(r)=\sqrt{1-r^2}$. Then $\tilde{b}(R_t)=b(R_t)$ and $\tilde{\sigma}(R_t)=\sigma(R_t)$ on $(-1,1)$, and therefore, $R_t$ is an $(-1,1)$-valued weak solution of the SDE $dR_t=b(R_t)dt+\sigma(R_t)dW^{(3)}_t$. For the existence and uniqueness of strong solutions, we note that $b(.)$ is Lipschitz continuous, and $\sigma(.)$ is H\"older continuous of order $1/2$. Hence, pathwise uniqueness holds for this SDE; see Theorem 3.5(ii) in \cite{revuz2013continuous}. As a result, any $(-1,1)$-valued solution is a strong solution by the Yamada--Watanabe theorem; see e.g.\ Theorem 1.7 in \cite{revuz2013continuous}. 
\end{proof}

Although our main focus in this paper is on pricing and hedging of long-term contracts, let us indicate how the framework can be adjusted to incorporate features that are important over shorter time horizons.

\subsection*{Negative prices}
In short-term electricity markets (real-time or day-ahead markets), prices frequently become negative; see e.g.\ \cite{carmona2014survey} for PJM, \cite{genoese2010occurrence} for German EEX. As electricity is non-storable, any disturbance of demand or of supply can cause negative prices.\footnote{To be more precise, negative prices can be caused by e.g.\ error predictions of the load, high temperature volatilities, network transmission and congestion issues (causing oversupply in one region and undersupply in another), and overdemand through prediction error from generation via renewable energy (wind and PV).} The polynomial model can be extended to allow for negative prices for short-term modeling by simply taking $c<0$. This way the spot price is bounded from below by $c$, $S_t\geq c$, which can be negative. This small modification does not change the polynomial structure, and thus, all computations and properties for forwards and hedges remain the same.

For long-term markets this feature is less relevant, as long-term prices are generally insensitive to temporary shocks. Indeed, the data of German Calendar year baseload forwards (over 8 years) does not contain negative prices.

\subsection*{Seasonality}
In electricity markets, prices highly depend on the exact delivery period, e.g.\ offpeak vs.\ peak hours, winter months vs.\ summer months, or specific quarters. Thus, if we compare contracts with same delivery length but different delivery periods, that is, different subperiods of a year, it is important to first adjust for seasonality before making reasonable comparison.  It is possible to incorporate seasonality by making $p_S$ not only a state-dependent, but also time-dependent mapping. More specifically, we can let $p_S(t,x) := c(t)+x^\top Q(t) x$, where $c$ and $Q$ have temporal components. This leads to a time-inhomogeneous version of the polynomial property, which remains tractable.
 
Note that all yearly baseload contracts deliver throughout the year and not only for a specific subperiod of the year. To capture these forwards in long-term markets, it is not necessary to explicitly model seasonality.

\subsection*{Spikes/Jumps} 
In short-term markets, one often observes extreme price changes in spot prices, known as spikes. These result from unanticipated shocks in demand, and exist only temporarily. In other words, prices don't stay at the new level, but revert rapidly back to the previous level. Because of their temporary nature, it is reasonable to argue that the spikes have a negligible effect on long-term prices, and therefore, should not be included in the framework for modeling long-term electricity forwards. 

However, our model can be extended to account for spikes if needed, say to model short-term spot prices, or joint short- and long-term markets. One possible way of doing so is to multiply the spot price by a mean-reverting jump process that jumps and then very quickly mean-reverts towards its standard level of $1$. A simple example is given by:
\begin{align*}
S_T &= p_S(X_t) J_t,\\
dJ_t &= \theta_J(1-J_t)dt + \int \sigma_J(X_t, v) N(dv ,dt), 
\end{align*}
where $N(dv,dt)$ is a Poisson random measure, $\theta_J$ a large mean-reversion parameter, which forces the process to revert quickly to the previous level after a jump.
Another possibility is to incorporate spikes by an additive component, e.g.
\begin{align*}
S_T &= p_S(X_t) +J_t,\\
dJ_t &= -\theta_J J_tdt +  \int \sigma_J(X_t, v)N(dv, dt).
\end{align*}
Either way, the extensions do not change the behavior of long-dated forward but only the short-term forward and spot, because all the jumps mean-revert very quickly and so do not have an effect on long term prices. Provided $\sigma_J(x,v)$ is chosen appropriately, many of the properties of polynomial diffusions (such as the moment formula) still apply; see \cite[Section~5]{filipovic2019polynomial} for more details.

\section{Polynomial diffusions and moment formulas}

In this section we briefly review some important results regarding polynomial diffusions; see e.g.\ \cite{Larsson2016} for more details. We also provide a moment formula for polynomials of degree two. Consider the (extended) generator $\G$ associated to the $\R^d$-valued polynomial diffusion $X_t$ introduced in \sref{Section}{sec_model}, namely
$$\G f(x) = \dfrac{1}{2}\text{Tr}(a(x)\nabla^2 f(x)) + (\kappa(\theta-x))^\top \nabla f(x)$$
 for $x\in\R^d$ and any $C^2$ function $f$. By It\^{o}'s formula, the process 
$$f(X_t) -f(X_0)-\int_0^t \G f(X_u) du$$ 
is a local martingale. Since the components of $a(x)$ are polynomials of degree at most two, it follows that for any $n\in\mathbb{N}$ and any polynomial $p\in \text{Pol}_n$, $\G p$ is also polynomial of the same degree or lower degree, i.e.\ $\G p \in \text{Pol}_n$.

Fix $n$ and let $N=\binom{d+n}{n}$ be the dimension of $\text{Pol}_n$. Let $H:\R^d \rightarrow \R^N$ be a function whose components form a basis of $\text{Pol}_n$. Then for any $p \in \text{Pol}_n $ 
\begin{align}
p(x) &= H(x)^\top  \vec{p},
\label{p_repres}\\
\G p(x)&= H(x)^\top  G \, \vec{p},
 \label{Gp_repres}
\end{align}
where $\vec{p}\in \R^N$ is the coordinate representation of $p(x)$, and $G\in \R^{N\times N}$ the matrix representation of the generator $\G$.

\begin{thm}[Moment formula - general version]\label{Thm_Momentformula1}
Let $p$ be a polynomial with coordinate representations \eqref{p_repres}--\eqref{Gp_repres}. Further let $X_t$ satisfy \eqref{X_SDE}. Then for $0\leq t\leq T$ we have: 
\begin{align}
 \E_\Q[ p(X_T) | ~ \mathcal{F}_t] = H(X_t)^\top  e^{(T - t)G} \vec{p}. 
\end{align}

\end{thm}
\begin{proof}
See Theorem 3.1 of \cite{Larsson2016}.
\end{proof}
Below we give a more explicit version of the moment formula for polynomials of degree two. Note that the quantity ${\rm tr}(\pi \,a(x))$ with $\pi\in\S^d$ is quadratic in $x$, and thus of the form
\begin{align}
{\rm tr}(\pi \,a(x)) = a_0(\pi) + a_1(\pi)^\top x + x^\top a_2(\pi) \,x.
\label{eq_trace}
\end{align}
for some $a_0(\pi)\in\R$, $a_1(\pi)\in\R^d$, and $a_2(\pi)\in\S^d$ that depend linearly on $\pi$.

\begin{thm}[Moment formula for polynomials of degree two]\label{Thm_Momentformula2}
Let $q(x)$ be a polynomial of the form $q(x)= q_0 + \vec{q}^\top x +x^\top Q x$ with $q_0 \in \R$, $\vec{q}\in \R^d$ and $Q\in \S^d$. Further let $X_t$ satisfy \eqref{X_SDE}. Then for $0\leq t\leq T$ we have: 
\begin{align*}
 \E_\Q[~ q(X_T) ~ | ~ \mathcal{F}_t] = \phi(T-t) + \psi(T-t)^\top X_t + X_t^\top \pi(T-t)X_t, 
 \end{align*}
where $\phi,\psi,\pi$ solve the linear ODE
\begin{equation}\label{eq_poly}
\begin{aligned}
\phi^{\prime} &= \psi^\top \kappa \theta + a_0(\pi),  && \phi(0) =  q_0, \\
\psi^{\prime} &= -\kappa^\top \psi + 2\pi\kappa\theta + a_1(\pi), && \psi(0) = \vec{q},\\
\pi^{\prime} &= -\pi \kappa -\kappa^\top \pi  + a_2(\pi), && \pi(0) = Q.
\end{aligned}
\end{equation}
\end{thm}

\begin{proof}
Define 
\begin{align*}
M(t, X_t) := \phi(T-t) + \psi(T-t)^\top X_t + X^\top_t \pi(T-t)X_t.
\end{align*}
Let $\tau = T-t$.  It\^o's formula along with \eqref{eq_trace} and then \eqref{eq_poly} gives:
\begin{align*}
dM(t, X_t) &= -\phi^\prime(\tau)dt -\psi^\prime(\tau)^\top X_t dt - X_t^\top \pi^\prime(\tau) X_t dt +\psi(\tau)^\top dX_t \\
&\quad + 2X_t^\top \pi(\tau) dX_t +\frac{1}{2}\cdot 2{\rm tr}(\pi(\tau)d\langle X\rangle_t) \\
&= \Bigl( -\phi^\prime(\tau) -\psi^\prime(\tau)^\top X_t  - X_t^\top \pi^\prime(\tau)X_t  +\psi(\tau)^\top \kappa \theta - \psi(\tau)^\top \kappa X_t \\
&\quad +2\theta^\top\kappa^\top\pi(\tau)X_t -  2 X_t^\top \pi(\tau) \kappa X_t + {\rm tr}(\pi(\tau)a(X_t)) \Bigl)dt \\
&\quad +\widehat{\sigma}(t, X_t)dW_t \\
&= \Bigl( \left[ -\phi^\prime(\tau) +\psi(\tau)^\top \kappa\theta +a_0(\pi(\tau))\right]\\ 
&\quad + \left[ -\psi^\prime(\tau) - \kappa^\top\psi(\tau) + 2\pi(\tau)\kappa\theta + a_1(\pi(\tau))\right]^\top X_t\\ 
&\quad+X_t^\top \left[ -\pi^\prime(\tau)  - \pi(\tau)\kappa - \kappa^\top\pi(\tau) + a_2(\pi(\tau)) \right]  X_t\Bigr) dt +\widehat{\sigma}( t,X_t)dW_t\\
&=  \widehat{\sigma}(t,X_t)dW_t,
\end{align*}
where $\widehat{\sigma }(t,x) := (\psi(\tau) + 2\pi(\tau)x)^\top \sigma( x )$. Thus, $M(t,X_t)$ is a local martingale. Now we let $C\in \R$ be a constant such that $\|a(x)\|_{\rm op}\leq C(1+\|x\|^2)$. Then with Cauchy-Schwartz inequality,
\begin{align*}
\|\widehat{\sigma}(t, X_t) \|^2 
&\leq \|\psi(T-t) + 2\pi(T-t)X_t \|^2  \| a(X_t)\|_{\rm op}\\
&\leq \widetilde{C} (1+\|X_t\|^4),
\end{align*}
for some constant $\widetilde{C}\in\R$. Together with Tonelli's theorem, this bound yields
\[
\E\left[ \int_0^T \|\widehat{\sigma}(t, X_t) \|^2dt\right] \leq  \widetilde{C} \int_0^T  \E\left[ 1+\|X_t\|^4 \right] dt,
\]
which is finite by \sref{Theorem}{Thm_Momentformula1}. Hence, $M(t,X_t)$ is a square-integrable true martingale. As a result,
 \[ M(t,X_t) = \E [M(T,X_T)|\mathcal{F}_t]=\E[q_0 + \vec{q} X_T+ X^\top_TQ X_T|\mathcal{F}_t]=\E[q(X_T)|\mathcal{F}_t].\]
This is the claimed formula.
\end{proof}

\section{The term structure of forward prices}\label{s_forwardprices}
In this section we define electricity forwards, present their pricing formulas and give expressions for covariances and correlations between different forwards. 

The price at time-$t$ of an electricity forward with instantaneous delivery at time $T\geq t$ is given by
\begin{align}
f(t, T, X_t) :=  \E_\Q\left[  S_T \big|  \mathcal{F}_t\right] .
\label{Eq_fwd_def_instantaneous}\end{align}
In practice, electricity is not delivered instantaneously, but gradually over a period of time. This leads us to the following definition: the time-$t$ price of an electricity forward with delivery period $[T_1, T_2)$, $t\leq T_1<T_2$, is given by 
 \begin{align}
F(t, T_1, T_2, X_t) :=  \frac{1}{T_2 - T_1}\E_\Q\left[ \int_{T_1}^{T_2} S_u \, du\big|  \mathcal{F}_t\right] .\label{Eq_fwd_def}\end{align}

Note that a forward contract (financial or physical) can have settlement that takes place either before or after the delivery period. Discounting is not needed in the pricing, as the difference in cashflow can be evened out by the purchase of a bond of that time period. $F(t,T_1,T_2, X_t)$ is often also referred to as swap price, as the delivery of underlying power happens over a period of time and thus the price is the averaged price over that period. 

It is intuitive that a forward with delivery period is the summation of all forwards (with instantaneous delivery) that deliver at single time points within the delivery period; moreover, a forward with delivery period which collapses into one single time point should be priced the same as a forward with instantaneous delivery. The following proposition confirms this relationship between forwards with and without delivery period. 
\begin{prop}\label{prop_Ff}
For $t\leq T_1 \leq T_2$, we have:
\begin{align*}
F(t,T_1, T_2,X_t) = \frac{1}{T_2 - T_1} \int_{T_1}^{T_2} f(t, u, X_t)du. 
\end{align*}
Moreover, 
\begin{align*}
\lim_{T_2 \rightarrow T_1 }F(t, T_1, T_2, X_t) = f(t, T_1, X_t).
\end{align*}
\end{prop}

\begin{proof}
In view of \eqref{Eq_fwd_def_instantaneous} and \eqref{Eq_fwd_def}, the first identity follows from the conditional version of Tonelli's theorem since $S_t$ is nonnegative. The second identity then follows from the fundamental theorem of calculus, using that $f(t,T,X_t)$ is continuous in $T$, see \sref{Proposition}{prop_fwd} below.
\end{proof}

The following result gives closed-form expression for the forward prices. 

\begin{prop}[{Pricing formula for forwards}]\label{prop_fwd}
 Let $\vec p_S$  be the coordinate representation of $p_S(x)$. The time-$t$ price of $f(t,T, X_t)$ for $t\leq T$ is
\begin{align*}
f(t,T, X_t) = H(X_t)^\top  e^{(T - t)G} \vec p_S,
\end{align*}
and the time-$t$ price of $F(t,T_1, T_2, X_t)$ for $t\leq T_1\leq T_2$ is
\begin{align*}
F(t, T_1, T_2, X_t) = \frac{1}{T_2 - T_1} H(X_t)^\top  e^{(T_1 - t)G} \int_{0}^{T_2 -T_1}e^{uG} du \, \vec p_S. 
\end{align*}
\end{prop}
\begin{proof}
This follows from \sref{Theorem}{Thm_Momentformula1} and rearranging terms.
\end{proof}

Note that $G$ is a non-invertible matrix. Still, $\int^{\tau}_0 e^{uG}du$ is explicit; see \sref{Appendix}{Appen_eGs} for the explicit computation. 

\medskip
\paragraph{\sref{Specification}{Sp_1}} Recall the \sref{Specification}{Sp_1} in \sref{Section}{sec_model}. We consider the basis given by
\begin{align}
H(x)= (1,z,y,z^2,yz,y^2 )^\top, ~ x=(z,y)^\top.
\label{Eq_H1}
\end{align}
Then $S_t$ can be uniquely represented as: 
\begin{align}
S_t= H(X_t)^\top \vec p_S ~ \text{ with }~\vec p_S = (c, 0,0, \beta,0,\alpha)^\top.
\label{Eq_pS1}
\end{align} 
For any $C^2$ function $f$ and $x= (z, y)^\top\in \R^2$, the generator $\G$ is :
\begin{align*}
\G f(x)
 = \begin{pmatrix}
  -\kappa_Z  z\\ \kappa_Y  z - \kappa_Y  y 
\end{pmatrix}^\top  \nabla f(x) + \frac{1}{2}\text{Tr}\left( 
\begin{pmatrix}
\sigma^2_Z & \rho\sigma_Y\sigma_Z \\
\rho\sigma_Y\sigma_Z & \sigma^2_Y
\end{pmatrix}
 \nabla^2 f(x)\right) .
\end{align*}
Applying $\mathcal{G}$ to each element of $H(X_t)$ gives its matrix representation,
 \begin{align}
G =
	\begin{pmatrix}
	0 & 0 & 0 & \sigma^2_Z & \rho\sigma_Y\sigma_Z & \sigma^2_Y \\
	0 & -\kappa_Z & \kappa_Y & 0 & 0 & 0 \\
	0 & 0 & -\kappa_Y & 0 & 0 & 0 \\
	0 & 0 & 0 & -2\kappa_Z & \kappa_Y & 0 \\
	0 & 0 & 0 & 0 & -\kappa_Z - \kappa_Y & 2\kappa_Y \\
	0 & 0 & 0 & 0 & 0 & -2\kappa_Y  
	\end{pmatrix} .
	\label{Eq_G1}
\end{align}

\medskip
\paragraph{\sref{Specification}{Sp_2}}Recall \sref{Specification}{Sp_2} in \sref{Section}{sec_model}. Here the general $\G$ actually preserves a proper subspace of $\text{Pol}_2$, namely the one spanned by the components of 
\begin{align}
H(x) =(1, z,y,r,z^2,yz,y^2)^\top, ~ x=(z,y,r)^\top. \label{Eq_H2} 
\end{align}
Therefore it is not necessary to include the remaining basis functions in the definition of $H$.
Then $S_t$ can be uniquely represented as
\begin{align}
S_t= H(X_t)^\top \vec p_S ~ \text{ with }~\vec p_S= (c, ~0,~0,~0,~\beta, ~0, ~\alpha)^\top.
\label{Eq_pS2}
\end{align}
For any $C^2$ function $f$ and $x= (z,y,r)^\top\in \R^3$,  the generator $\G$ is 
\begin{align*}
\G f(x) = \begin{pmatrix}
 -\kappa_Z  z\\
-\kappa_Y  y + \kappa_Y  z \\ 
 \kappa_R(\theta_R -r)
\end{pmatrix}^\top  \nabla f(x) 
+ \frac{1}{2}\text{Tr}\left( 
\begin{pmatrix}
\sigma^2_Z & \sigma_Y\sigma_Z  r & 0 \\
\sigma_Y\sigma_Z  r & \sigma^2_Y & 0\\
0  & 0 & \sigma^2_R (1- r^2)
\end{pmatrix}
 \nabla^2 f(x)\right) .
\end{align*}
Applying $\mathcal{G}$ to each element of $H(X_t)$ gives
\begin{align}
G =
	\begin{pmatrix}
	0 & 0 & 0 & \kappa_R \theta_R  &\sigma^2_Z & 0 & \sigma^2_Y \\
	0 & -\kappa_Z & \kappa_Y & 0& 0 & 0 & 0 \\
	0 & 0 & -\kappa_Y &0& 0 & 0 & 0 \\
	0 & 0 & 0 &-\kappa_R & 0& \sigma_Y\sigma_Z & 0 \\
	0 & 0 & 0 &0 & -2\kappa_Z & \kappa_Y & 0 \\
	0 & 0 & 0&0 & 0 & -\kappa_Z - \kappa_Y & 2\kappa_Y \\
	0 & 0 & 0& 0& 0 & 0 & -2\kappa_Y  
	\end{pmatrix} .
	\label{Eq_G2}
\end{align}

For later use, we briefly discuss the instantaneous quadratic covariation and correlations between different forwards and give explicit forms for both specifications. The instantaneous covariation between two forwards with instantaneous delivery at $T_1$ and $T_2$ is, at time $t\le T_1\wedge T_2$, 
\begin{align}
\frac{d}{dt}\langle f(t, T_1, X_t), f(t, T_2, X_t)  \rangle
= \vec p_S^\top  {e^{(T_2-t)G}}^\top \Sigma (X_t)   e^{(T_1-t)G} \vec p_S,
\label{Eq_covwodel}
\end{align}
where 
\begin{align}  
 \Sigma (X_t)dt =d\langle H(X), H(X)\rangle_t. \label{Eq_sigma}
\end{align}
We define the corresponding instantaneous correlation as 
\begin{equation} \label{Eq_corwodel}
\begin{aligned}
 &\mathrm{Corr}[ f(t, T_1,  X_t), f(t, T_2, X_t)] \\
 &= \frac{\vec p_S^\top {e^{(T_2-t)G}}^\top \Sigma (X_t)    e^{(T_1-t)G} \vec p_S}{\sqrt{\vec p_S^\top  {e^{(T_1-t)G}}^\top \Sigma( X_t)  e^{(T_1-t)G} \vec p_S \, \vec p_S^\top  {e^{(T_2-t)G}}^\top \Sigma (X_t)  e^{(T_2-t)G} \vec p_S}}
 \end{aligned}
\end{equation}
with $\Sigma (X_t)$ from \eqref{Eq_sigma}.
The matrices $ \Sigma (X_t)$ for \sref{Specification}{Sp_1} and \sref{Specification}{Sp_2} are given in \sref{Appendix}{append_Sigma}. Similarly, for $t\leq T_1 < T_2$ and $t \leq T_3 < T_4$, the time-$t$ instantaneous covariation of forwards  with delivery periods $[T_1,T_2)$ and $[T_3,T_4)$ is
\begin{align}
\frac{d}{dt}\langle F(t, T_1, T_2, X_t), F(t, T_3, T_4, X_t)  \rangle
= \vec w_{34}^\top  {e^{(T_3-t)G}}^\top \Sigma(X_t)   e^{(T_1-t)G} \vec w_{12},
 \label{Eq_covwdel}
\end{align}
and the time-$t$ instantaneous correlation is: 
\begin{equation} \label{Eq_corwdel}
\begin{aligned}
 &\mathrm{Corr}[~ F(t, T_1, T_2, X_t), F(t, T_3, T_4, X_t)~ ] \\
 &= \frac{\vec w_{34}^\top {e^{(T_3-t)G}}^\top \Sigma (X_t)    e^{(T_1-t)G} \vec w_{12}}{\sqrt{\vec w_{12}^\top  {e^{(T_1-t)G}}^\top \Sigma (X_t)  e^{(T_1-t)G} \vec w_{12}\,\vec w_{34}^\top  {e^{(T_3-t)G}}^\top \Sigma (X_t) e^{(T_3-t)G} \vec w_{34}}}
 \end{aligned}
\end{equation}
with $\Sigma (X_t)$ from \eqref{Eq_sigma} and 
\begin{align}
\vec w_{ij} = \int_{T_i}^{T_j}e^{uG} du \,\vec p_S.
\label{Eq_w}
\end{align}

\begin{rem}[Option pricing]
Let $p(X_T)$ be the payoff function of an option based on a forward or a spot. For example, for a European call on a forward with delivery period $[T_1, T_2)$, strike $K$, and maturing $T$, we have $p(X_T) = (F(T,T_1, T_2, X_T)-K)^+$. Modulo discounting, the time-$t$ price of such an option is the $\mathcal{F}_t$-conditional expectation of $p(X_T)$ under $\Q$. If $p$ is a polynomial function, we can obtain explicit pricing for the option by \sref{Theorem}{Thm_Momentformula1} (if the option is based on a spot) or \sref{Proposition}{prop_fwd} (if the option is based on a forward). 
If $p$ is not a polynomial, an approximation scheme is required. For example, one can use the polynomial expansion method described in \cite[Section~7]{filipovic2019polynomial}.
\end{rem}

\section{Market price of risk specification}\label{sec_lambda}
In order to incorporate time series observations of real-world forward curves, we must specify the forward dynamics under the real-world probability measure $\mathbb{P}$. Thus, in this section, we specify a market price of risk function $\lambda: \mathbb{R}^d \rightarrow \mathbb{R}^{d}$ by
\begin{align*}
 \lambda(x)= \sigma(x)^{-1}(\gamma + \Lambda \,x)
\end{align*} 
for some $\gamma \in \R^d$ and $\Lambda\in \mathbb{S}^{d\times d}$, and denote the associated Radon--Nikodym density process by 
\begin{align}
M^\lambda_t=\exp{\Big( \,\int_0^t \lambda(X_s)^\top dW_s - \frac{1}{2}\int_0^t \|\lambda(X_s)\|^2 ds \,\Big)}.\label{Eq_E}
\end{align} 
We choose $\gamma$ and $\Lambda$ such that $M^\lambda_t$ is a true martingale. We can then define $\mathbb{P}$ on every finite time interval $[0, T]$ via its Radon--Nikodym density $\frac{d\mathbb{P}}{d\mathbb{Q}}|_{\mathcal{F}_T}=M^\lambda_T$. Then, by Girsanov's theorem, the $\mathbb{P}$-dynamics of $X_t$ becomes 
\begin{align}
dX_t &= [(\kappa \theta+ \gamma) - (\kappa - \Lambda)X_t]dt + \sigma(X_t)\, dW^\mathbb{P}_t \label{Eq_Pdynamic}
\end{align}
with 
$dW^\mathbb{P}_t := dW_t -\lambda(X_t) dt $. Note that the speed of mean reversion is now adjusted to $\kappa-\Lambda$ from $\kappa$. 

Consider now \sref{Specification}{Sp_1}. In this case $M^\lambda_t$ is a true martingale for any choice for $\gamma$ and $\Lambda$, as the following result shows.

\begin{prop} Let $X_t$ evolve according to \eqref{Eq_2factor}. Then $M^\lambda_t$ from \eqref{Eq_E} is a martingale. 
\end{prop}

\begin{proof}
Define 
$
\widetilde{X}_t:=(Z_t,Y_t,Z^2_t,Y_tZ_t,Y^2_t,\int_0^t (\sigma^{-1}\gamma+\sigma^{-1}\Lambda X_t)^\top dW_t)^\top$. Note that $\widetilde{X}_t$ has drift and diffusion that are affine in $\widetilde{X}_t$; see computations in \sref{Section}{s_forwardprices}, \sref{Appendix}{append_Sigma} and \sref{Section}{subsec_KF}. Thus, by Kallsen \& Muhle-Karbe (Corollary~3.9 in \cite{kallsen2010exponentially}), $M^\lambda_t$ is a true martingale.
\end{proof}

To be explicit, let $\Lambda= \diag(\lambda_{Z},\lambda_{Y})$ and $\gamma =(\gamma_Z, \gamma_Y)^\top$. Then the $\mathbb{P}$-dynamics of $X_t$ is given by:
\begin{align}
dX_t =&\left[\begin{pmatrix}
     \gamma_Z \\
       \gamma_Y 
    \end{pmatrix} 
-  \begin{pmatrix}
     \kappa_Z - \lambda_Z & 0 \\
       -\kappa_Y  & \kappa_Y- \lambda_Y 
    \end{pmatrix}
       X_t\right] dt 
      + \begin{pmatrix}
    \sigma_Z & 0 \\
    \rho\sigma_Y  & \sigma_Y\sqrt{1 - \rho^2}  \\
    \end{pmatrix}  dW^{\mathbb{P}}_t . \label{Eq_Sp_Pdynam}
\end{align}
This can also be written as $dX_t=\kappa^\prime(\theta^\prime-X_t)dt+\sigma(X_t)dW^{\mathbb{P}}_t$ with
\begin{align*}
\kappa^\prime &=  \begin{pmatrix}
       \kappa_Z - \lambda_Z & 0 \\
       -\kappa_Y  & \kappa_Y- \lambda_Y   
\end{pmatrix} , &
\theta^\prime =  \begin{pmatrix}
\frac{\gamma_Z}{\kappa_Z - \lambda_Z} \\
\frac{\gamma_Y}{\kappa_Y-\lambda_Y}+ \frac{\kappa_Y\gamma_Z}{(\kappa_Z-\lambda_Z)(\kappa_Y-\lambda_Y)}
   \end{pmatrix},
\end{align*}
and $\sigma(x)$ from \eqref{Eq_sigmax1}.

In the case of \sref{Specification}{Sp_2} it is a more delicate problem to determine those market price of risk parameters for which $M^\lambda_t$ is a true martingale. Since we will not use \sref{Specification}{Sp_2} in our empirical analysis, we do not consider this issue here.

\subsection*{Forward risk premium}

We define the forward risk premium as the difference of the forward and the predicted spot price. The time-$t$ forward risk premium of a forward with instantaneous delivery at $T\geq t$ is thus given by
\begin{align*}
R(t,T, X_t) :&= \E_\mathbb{Q}[S_T\mid\mathcal{F}_t] - \E_\mathbb{P}[S_T\mid\mathcal{F}_t],
\end{align*}
and the time-$t$ forward risk premium of a forward with delivery period $[T_1, T_2)$, $t\leq T_1<T_2$, is given by
\begin{align*}
R(t, T_1, T_2, X_t):&= \frac{1}{T_2-T_1}\E_\mathbb{Q}\left[\int_{T_1}^{T_2}S_u \,du\mid\mathcal{F}_t\right] - \frac{1}{T_2-T_1}\E_\mathbb{P}\left[\int_{T_1}^{T_2}S_u \,du\mid \mathcal{F}_t\right].
\end{align*}
The notion above is consistent with the ex-ante notion of forward risk premium used by e.g.\ \cite*{benth2008pricing, benth2009information, benth2012critical, benth2014pricing,benth2019mean, krevcar2019towards}. Both the $\Q$- and $\mathbb{P}$-conditional expectations can be computed using the pricing formula in \sref{Proposition}{prop_fwd}. We obtain the following explicit expressions for forward risk premia:  
\begin{align*}
R(t,T, X_t) &= H(X_t)^\top \big[ e^{(T - t)G}  -e^{(T - t)G^\lambda}  \big] \vec p_S , \\
\intertext{and}
R(t,T_1, T_2, X_t) 
&= \frac{1}{T_2-T_1}H(X_t)^\top \big[ e^{(T_1 - t)G} \int_{0}^{T_2 -T_1}e^{uG} du -e^{(T_1 - t)G^\lambda} \int_{0}^{T_2 -T_1}e^{uG^\lambda} du \big] \vec p_S,
\end{align*}
where $G^\lambda$ denotes the matrix representation of the generator $\G$ under $\mathbb{P}$. 
For example, for \sref{Specification}{Sp_1} under $\mathbb{P}$, $X_t$ evolves according to \eqref{Eq_Sp_Pdynam}, and $G^\lambda$ is given by   
\begin{align*}
G^\lambda =
	\begin{pmatrix}
	0 & \gamma_Z & \gamma_Y & \sigma^2_Z & \rho\sigma_Y\sigma_Z & \sigma^2_Y \\
	0 & \lambda_Z-\kappa_Z & \kappa_Y & 2\gamma_Z & \gamma_Y & 0 \\
	0 & 0 & \lambda_Y-\kappa_Y & 0 & \gamma_Z & 2\gamma_Y \\
	0 & 0 & 0 & 2(\lambda_Z-\kappa_Z) & \kappa_Y & 0 \\
	0 & 0 & 0 & 0 & (\lambda_Z+\lambda_Y)-(\kappa_Z + \kappa_Y) & 2\kappa_Y \\
	0 & 0 & 0 & 0 & 0 & 2(\lambda_Y-\kappa_Y)  
	\end{pmatrix} .
\end{align*}

The forward risk premium arises from the market price of risk $\lambda(X_t)$ and the associated measure change via the Girsanov's theorem, designed so that the polynomial structure is preserved. This produces stochastic and time varying forward risk premia. The risk premia do not have a definite sign, and can alternate between being positive and negative.\footnote{Empirical studies of electricity forward risk premia show mixed findings; see e.g.\ \cite{bunn2013forward} for a literature survey, and \cite{valitov2019risk, viehmann2011risk} for discussions of the risk premium in the short-term German market in particular.} There is an extensive literature on market price of risk specifications, forward risk premia, and measure changes for electricity modeling; see e.g\  \cite*{benth2008pricing, weron2008market, benth2009information, benth2014pricing, krevcar2019towards, benth2019mean}.

\section{Hedging} \label{sec_hedging}

In this section we first describe a rolling hedge setup with constraints which addresses the illiquidity and non-storability issues when hedging a long-term electricity contract. Rolling hedges for commodities form a well-known hedging scheme; see for example \cite{glasserman2001shortfall,neuberger1999hedging}. We then briefly review the locally risk-minimizing hedge of F\"ollmer and Schweizer, and give a rolling hedge for our modeling framework that is risk-minimizing.

\subsection{A rolling hedge setup}
Suppose we have committed to deliver power from year $\T$ to year $\T+1$ for a large $\T\in \N$ (e.g. $\T=10$ years) and our objective is to hedge this long-term electricity commitment. In our framework the time-$t$ valuation of the commitment is 
\begin{align*}
\F_t:= F(t, \T, \T+1, X_t)= \E_\Q\left[ \int_{\T}^{\T+1} S_u  du\big|  \mathcal{F}_t\right]
\end{align*}
Note that $\F_t$ is a $\Q$-martingale and the pricing formula (\sref{Proposition}{prop_fwd}) gives explicit pricing at any $t\in [0, \T]$.
 In an interest rate context, the analogous hedging problem is rather easy: just buy bonds and hold them as the payout in 10 years is known in advance. For electricity the problem is more difficult for a number of reasons:
\begin{itemize}
\item Long-term forwards are not liquidly traded (otherwise \emph{buy and hold} the financial contracts as in the interest rate context);
\item Electricity cannot be stored without significant costs (otherwise \emph{cash and carry} as for other storable commodities: simply buy the amount needed in $[\T,\T+1]$ and hold).
\item Only short-term / near-dated contracts with same delivery length is available. But its underlying commodity (electricity) is not the same as the one underlying a long-term contract because power is not storable. Some empirical evidence suggests that short-term prices carry limited information about what spot prices will be far into the future (see \cite{handika2012relationship}). 
\end{itemize}

One possible strategy in this case is a \emph{rolling hedge}, where we take a long position in near-term contracts as a hedge, and roll the hedge going forward. The underlying assumption of this strategy is that near-dated yearly contracts are highly correlated with far-dated yearly contracts, and become more so as the maturity date approaches\footnote{Note that this statement does not contradict the common perception that the short- and long-term data are not very correlated, e.g. \cite{koekebakker2005forward}. The first nearby calendar year forward is often considered a medium-term or even a long-term contract. }. This assumption is supported by the data; see \sref{Figure}{Fig_corr} in \sref{Appendix}{append_corr}.

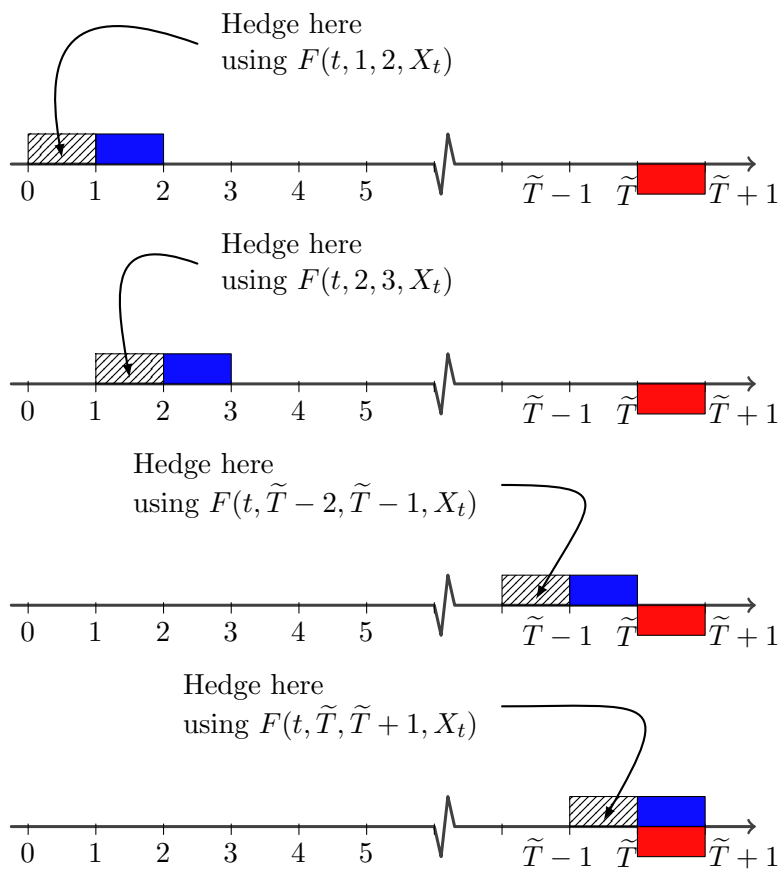
\begin{figure}[ht]
\begin{center}
\begin{tikzpicture}[line cap=round, line join=round, x=4.5cm, y=4cm, decoration={crosses}, set/.style={draw}]
  \draw[color=darkgray,line width=0.4mm,->] (-0.05,0) -- (1.2,0) -- (1.22,-0.1) -- (1.24,0.1) -- (1.26,0) -- (2.15,0);
  \foreach \x in {0,...,6}
     	\draw (\x/5,1pt) -- (\x/5,-3pt);
  \draw (1.4,1pt) -- (1.4,-3pt);
  \draw (1.6,1pt) -- (1.6,-3pt);
  \draw (1.8,1pt) -- (1.8,-3pt);
  \draw (2,1pt) -- (2,-3pt);
  \draw[color=black] (0.05,-10pt) node[left] {$0$};
  \draw[color=black] (0.25,-10pt) node[left] {$1$};
  \draw[color=black] (0.45,-10pt) node[left] {$2$};
  \draw[color=black] (0.65,-10pt) node[left] {$3$};
  \draw[color=black] (0.85,-10pt) node[left] {$4$};
  \draw[color=black] (1.05,-10pt) node[left] {$5$};
  \draw[color=black] (1.7,-10pt) node[left] {$\T-1$};
  \draw[color=black] (1.83,-10pt) node[left] {$\T$};
  \draw[color=black] (2.25,-10pt) node[left] {$\T+1$};

{
  \draw[-latex,thick] (0.5,0.4) node [right] {$\begin{array}{l}\text{Hedge here}\\\text{using $F(t,1,2, X_t)$}\end{array}$} to [out=160,in=100,looseness=1.8] (0.1,0.02);
  \draw [fill=blue,fill opacity=0.95] (0.2,0) rectangle (0.4,0.1);
  \draw [pattern=north east lines] (0,0) rectangle (0.2,0.1);
  }
 \draw[fill=red, fill opacity = 0.95] (1.8,0) rectangle (2.0,-0.1); 
  
\end{tikzpicture}

\begin{tikzpicture}[line cap=round, line join=round, x=4.5cm, y=4cm, decoration={crosses}, set/.style={draw}]
  \draw[color=darkgray,line width=0.4mm,->] (-0.05,0) -- (1.2,0) -- (1.22,-0.1) -- (1.24,0.1) -- (1.26,0) -- (2.15,0);
  \foreach \x in {0,...,6}
     	\draw (\x/5,1pt) -- (\x/5,-3pt);
  \draw (1.4,1pt) -- (1.4,-3pt);
  \draw (1.6,1pt) -- (1.6,-3pt);
  \draw (1.8,1pt) -- (1.8,-3pt);
  \draw (2,1pt) -- (2,-3pt);

  \draw[color=black] (0.05,-10pt) node[left] {$0$};
  \draw[color=black] (0.25,-10pt) node[left] {$1$};
  \draw[color=black] (0.45,-10pt) node[left] {$2$};
  \draw[color=black] (0.65,-10pt) node[left] {$3$};
  \draw[color=black] (0.85,-10pt) node[left] {$4$};
  \draw[color=black] (1.05,-10pt) node[left] {$5$};
  \draw[color=black] (1.70,-10pt) node[left] {$\T-1$};
  \draw[color=black] (1.83,-10pt) node[left] {$\T$};
  \draw[color=black] (2.25,-10pt) node[left] {$\T+1$};

{
  \draw[-latex,thick] (0.5,0.4) node [right] {$\begin{array}{l}\text{Hedge here}\\\text{using $F(t,2,3, X_t)$}\end{array}$} to [out=160,in=100,looseness=1.8] (0.3,0.02);
  \draw [pattern=north east lines] (0.2,0) rectangle (0.4,0.1);
  \draw [fill=blue,fill opacity=0.95] (0.4,0) rectangle (0.6,0.1);
\draw[fill=red, fill opacity = 0.95] (1.8,0) rectangle (2.0,-0.1); 
 
  }
\end{tikzpicture}
\begin{tikzpicture}[line cap=round, line join=round, x=4.5cm, y=4cm, decoration={crosses}, set/.style={draw}]
  \draw[color=darkgray,line width=0.4mm,->] (-0.05,0) -- (1.2,0) -- (1.22,-0.1) -- (1.24,0.1) -- (1.26,0) -- (2.15,0);
  \foreach \x in {0,...,6}
     	\draw (\x/5,1pt) -- (\x/5,-3pt);
  \draw (1.4,1pt) -- (1.4,-3pt);
  \draw (1.6,1pt) -- (1.6,-3pt);
  \draw (1.8,1pt) -- (1.8,-3pt);
  \draw (2,1pt) -- (2,-3pt);
  \draw[color=black] (0.05,-10pt) node[left] {$0$};
  \draw[color=black] (0.25,-10pt) node[left] {$1$};
  \draw[color=black] (0.45,-10pt) node[left] {$2$};
  \draw[color=black] (0.65,-10pt) node[left] {$3$};
  \draw[color=black] (0.85,-10pt) node[left] {$4$};
  \draw[color=black] (1.05,-10pt) node[left] {$5$};
  \draw[color=black] (1.70,-10pt) node[left] {$\T-1$};
  \draw[color=black] (1.83,-10pt) node[left] {$\T$};
  \draw[color=black] (2.25,-10pt) node[left] {$\T+1$};

{
  \draw[-latex,thick] (1.4,0.4) node [left] {$\begin{array}{l}\text{Hedge here}\\\text{using $F(t,\T-2,\T-1, X_t)$}\end{array}$} to [out=0,in=60,looseness=2.5] (1.5,0.02);
  \draw [pattern=north east lines] (1.4,0) rectangle (1.6,0.1);
  \draw [fill=blue, fill opacity=0.95] (1.6,0) rectangle (1.8,0.1);
  \draw[fill=red, fill opacity = 0.95] (1.8,0) rectangle (2.0,-0.1); 

  }
\end{tikzpicture}

\begin{tikzpicture}[line cap=round, line join=round, x=4.5cm, y=4cm, decoration={crosses}, set/.style={draw}]
  \draw[color=darkgray,line width=0.4mm,->] (-0.05,0) -- (1.2,0) -- (1.22,-0.1) -- (1.24,0.1) -- (1.26,0) -- (2.15,0);
  \foreach \x in {0,...,6}
     	\draw (\x/5,1pt) -- (\x/5,-3pt);
  \draw (1.4,1pt) -- (1.4,-3pt);
  \draw (1.6,1pt) -- (1.6,-3pt);
  \draw (1.8,1pt) -- (1.8,-3pt);
  \draw (2,1pt) -- (2,-3pt);
  \draw[color=black] (0.05,-10pt) node[left] {$0$};
  \draw[color=black] (0.25,-10pt) node[left] {$1$};
  \draw[color=black] (0.45,-10pt) node[left] {$2$};
  \draw[color=black] (0.65,-10pt) node[left] {$3$};
  \draw[color=black] (0.85,-10pt) node[left] {$4$};
  \draw[color=black] (1.05,-10pt) node[left] {$5$};
  \draw[color=black] (1.70,-10pt) node[left] {$\T-1$};
  \draw[color=black] (1.83,-10pt) node[left] {$\T$};
  \draw[color=black] (2.25,-10pt) node[left] {$\T+1$};

{
  \draw[-latex,thick] (1.4,0.4) node [left] {$\begin{array}{l}\text{Hedge here}\\\text{using $F(t,\T,\T+1, X_t)$}\end{array}$} to [out=0,in=60,looseness=2.5] (1.7,0.02);
  \draw [pattern=north east lines] (1.6,0) rectangle (1.8,0.1);
  \draw [fill=blue, fill opacity=0.95] (1.8,0) rectangle (2,0.1);
\draw[fill=red, fill opacity = 0.95] (1.8,0) rectangle (2.0,-0.1); 

  }
\end{tikzpicture}
\end{center}
\caption{ ~(color online). The mechanism of rolling hedges.}
\label{Fig_rollH}
\end{figure}

To formalize this, let us first define the price process $P_t$ containing all calendar-year forwards with a one-year delivery period (short: cal forward):
\begin{align}
P_t = 
\begin{pmatrix}[1.1]
P^1_t \\P^2_t\\ \vdots\\ P^{N-1}_t\\P^N_t
\end{pmatrix}
=\begin{pmatrix}[1.1]
F(t, 1, 2, X_t)\\
F(t, 2, 3, X_t)\\
\vdots\\
F(t, N-1, {N}, X_t)\\
F(t, N, N+1, X_t)\\
\end{pmatrix}
\label{Eq_Pt}
\end{align}
where $N=\T$ and $P^N_t= \F_t$.
Note that each $P^k_t$ is a $\Q$-martingale by its definition \eqref{Eq_fwd_def}. By \sref{Proposition}{prop_fwd}, $P^k_t$ can be expressed as
\begin{align}
P^k_t = H(X_t)^\top e^{(k-t)\, G}  \vec w_{01},\label{Eq_Pt1}
\end{align}
where $\vec w_{01}$ is defined in \eqref{Eq_w}.

An admissible\footnote{Note that for any polynomial processes $p(X_t)$ all moments of $P_t:= \E_\Q[p(X_T)|\mathcal{F}_t]$ exist. Therefore, integration with respect to any moments of $P$ is well-defined. And thus, $\xi\in L^2(P)$, i.e.\ $ \E_\Q[ \int_0^T ~\xi_s^\top d\langle P \rangle_s \xi_s ]< \infty$, and $\varphi := (\eta, \xi)^\top$ is admissible.  
} hedging strategy is an $\R^{N+1}$-valued process $\varphi_t= (\eta_t, \xi_t)^\top=(\eta_t, \xi^1_t,\dots , \xi^{N}_t)^\top $, where $\eta_t$ is adapted (representing bank account) and $\xi_t$ is predictable (representing amount of tradable assets or hedge ratio), and satisfies
\begin{align}
\xi^i_t = 0 ~~~~ \forall t\notin [ k-1,k ), ~k = 1, \dots , N.\label{Eq_H}
\end{align}
The constraint \eqref{Eq_H} reflects the liquidity issue and trading rule of those markets: 
\begin{itemize}
	\item only the first-nearby forwards are liquid;
	\item a contract that has started to deliver can no longer be traded. 
\end{itemize}
The value process (or the {P \& L}) at time $t\in [k-1,k )$ for $k \in \{1,...,N\}$ is 
\begin{align*}
V_t(\varphi) 
=\eta_t  + \xi_t^\top P_t = \eta_t + \xi^k_t  P^k_t
= \eta_t + \xi^k_t  F(t, k, k+1, X_t).
\end{align*}
The cumulative cost of the hedge up to time $t$ is: 
\begin{align*}
C_t(\varphi) := V_t(\varphi) - G_t(\varphi),
\end{align*}
where $G_t$ denotes the {cumulative gain} of the hedge up to time $t$:
\begin{align}
G_t(\varphi)&= \int_0^t \xi_s^\top dP_s 
= \sum_{i=1}^{k-1} \int_{i-1}^{i} \xi^i_s dP^i_s + \int_{k-1}^{t} \xi^k_s dP^k_s \label{Eq_gain}\\
&= \sum_{i=1}^{k-1} \int_{i-1}^{i} \xi^i_s dF(s, i,i+1,X_s) + \int_{k-1}^{t} \xi^k_s dF(s, k,k+1,X_s)  \nonumber
\end{align}
for $ t\in[k-1,k)$.

 Note that the market is incomplete under the restriction \eqref{Eq_H}, since there are two different Brownian motions, but only one risky asset to invest in at any given time. In an incomplete market a claim generally cannot be fully replicated at maturity by a self-financing hedging strategy. Depending on the restriction on cash account $\eta$, one can either use a strategy that is self-financing but does not perfectly replicate the claim at maturity, or use a strategy that fully replicates the claim at maturity but needs additional investment throughout the hedge, i.e.\ is not self-financing. In the first case, we have residual risk and in the latter case additional cash infusion is needed. Either way, risk cannot be fully eliminated and can only be minimized. In the following we briefly review the concept of risk-minimizing strategy in the sense of F\"ollmer and Schweizer, and then give a rolling hedge that is locally risk-minimizing.
  
\subsection{A locally risk minimizing hedging criterion}

The risk-minimization criterion proposed and developed by F\"ollmer and Schweizer (see e.g.\ \cite{heath2001comparison}, \cite{heath1999}, \cite{schweizer1999guided}, \cite{schweizer1990risk}, \cite{follmer1991hedging} for details), is to minimize the conditional variance $R_t(\varphi)$ of the cost process $C_t(\varphi)$,
\begin{align*}
R_t(\varphi) := \E_\Q\left[ ( C_T(\varphi) - C_t(\varphi) )^2 | \mathcal{F}_t \right],
\end{align*} 
among all not necessarily self-financing strategies $\varphi$ that perfectly replicate $\F$ at maturity:
\begin{align}\label{Eq_LRM1}
 V_\T(\varphi) = \F \quad \text{$\Q$-a.s.}
\end{align}
In our setup, \eqref{Eq_LRM1} is equivalent to $\eta_{T_N}=0$ and $\xi^{N}_{T_N}=1$.

A strategy $\varphi^*$ is called \emph{risk-minimizing} if for any $\varphi$ that satisfies \eqref{Eq_LRM1} we have $R_t(\varphi^*) \leq R_t(\varphi)$, $\Q$-a.s.\ for every $t\in [0, \T]$; see Schweizer (page 545 in \cite{schweizer1990risk}). One can show that any risk-minimizing strategy is mean self-financing, i.e.\ $C_t(\varphi)$ is a $\Q$-martingale. F\"ollmer and Schweizer showed that the existence of such a strategy $\varphi$ is guaranteed if the price process $P_t$ is a $\Q$-local martingale. Moreover, in the martingale case, finding such a strategy is equivalent to finding the Galtchouk--Kunita--Watanabe (GKW) decomposition of $\F$, namely
\begin{align}
\F=\E[\F]+ \int_0^{\T} \xit^\top_s dP_s + \Lt_{\T},\label{Eq_decom}
\end{align}
where $\xit$ is an admissible, predictable process and $\Lt$ is a square-integrable $\Q$-martingale strongly orthogonal to $P$ with $\Lt_0 = 0$. The risk-minimizing hedging strategy $\varphi^{rm}$ is then given by  
\begin{align*}\varphi^{rm}_t=\left(\eta^{rm}_t, \xi^{rm}_t\right)^\top= 
\left( V_t(\varphi^{rm}) -  \xi_t^{{rm}^\top}P_t,\, \xit_t\right)^\top ,
\end{align*}
where the value process is $V_t(\varphi^{rm}) =  \E[\F|\mathcal{F}_t]= \F_t =
\F_0+ \int_0^{t} \xi^{{\F}^\top}_s dP_s + \Lt_{t}$ and the cost process is $C_t(\varphi^{rm})= \F_0 +\Lt_t $. Obviously this risk-minimizing strategy satisfies $V_{\T}(\varphi^{rm}) = \F_{\T}$, and the associated risk process  $R_t(\varphi^{rm})$ is minimal (zero) at $t=\T$.

\subsection{A risk-minimizing rolling hedge}\label{sec:RM-hedge}
Recall that the price process $P_t$ is a $\Q$-martingale. Then the time-$t$ valuation of the long-term electricity commitment $\F_\T$ has a GKW-decomposition as in \eqref{Eq_decom}. We now compute the process $\xit$ in this decomposition. This will give us the hedging strategy. Using \eqref{Eq_decom}, \eqref{Eq_H} and \eqref{Eq_gain}, we obtain for any $t\in [k-1, k)$, $k\in\N$: 
\[
\langle P^k,  \F \rangle_t - \langle P^k,  \F \rangle_{k-1} = \int_{k-1}^t d\langle P^k, \int_{k-1}^{\Bigcdot}\xit^{k}_s  dP^k_{s}\rangle_u + \int_{k-1}^t d\langle P^k, \Lt \rangle_u =  \int_{k-1}^{t}\xit^{k}_s d\langle P^k, P^k \rangle_s,
\]
where $\langle P^k, \Lt \rangle_t=0$ as $\widetilde L$ is orthogonal to $P$, and $\langle P^k,  \F \rangle_{k-1}=0$ as $\F_{k-1}$ is constant and known at $t\geq k-1$.
Rearranging and using \eqref{Eq_sigma} and \eqref{Eq_w} we get the $k$-th component of $\xit_t$ for $t\in [k-1, k)$:
\begin{align}
\xit^{k}_t &=\frac{d\langle P^k, \F \rangle_t}{d\langle P^k, P^k \rangle_t}\nonumber \\
&= \frac{\vec w_{01}^\top  {e^{(\T-t) G}}^\top d\langle H(X),H(X)\rangle_t \,e^{(k-t)G} \vec w_{01}}{\vec w_{01}^\top {e^{(k-t) G}}^\top d\langle H(X),H(X)\rangle_t  \,e^{(k-t) G} \vec w_{01}} \nonumber\\
&=\frac{\vec w_{01}^\top {e^{(\T-t) G}}^\top \Sigma(X_t) \, e^{(k-t)G} \vec w_{01}}{\vec w_{01}^\top  {e^{(k-t) G}}^\top \Sigma(X_t) \, e^{(k-t) G} \vec w_{01}} .\nonumber
\end{align}
Therefore, the risk-minimizing hedging strategy of the tradable assets is given by
\begin{align}
\xi^{rm}_t &= \left( \xi_t^{rm, 1}, ..., \xi_t^{rm, N} \right)^\top, \label{Eq_xi}\\
\intertext{where} 
\xi^{rm,k}_t &= 
\begin{cases}
 \xit^{k}_t , ~~&\text{for }~ t\in[k-1,k);\\
  0, ~~~ &\text{otherwise. }
\end{cases}\label{Eq_xik}
\end{align}
And thus, for $t\in[k-1,k)$, the cash account $\eta^{rm}_t$ is then given by
\begin{align*}
\eta^{rm}_t = V_t(\varphi^{rm}) - \xi^{rm^\top}_t P_t = \F_t - \xit^{k}_t P^k_t,
\end{align*} 
and the associated cost process is 
\begin{align*} 
C_t(\varphi^{rm}) = \F_t- \int_0^t \xi^{{rm}^\top}dP_s. 
\end{align*}

\begin{rem}
The risk minimizing strategy also minimizes the quadratic covariation between the claim and the value of hedge without the cash account. Indeed, formally one has
\begin{align*}
	 \min_{\xi} d\langle \F - \xi^k P^k \rangle_t
   = \min_{\xi} \left( d\langle \F \rangle_t - 2\xi^k_t d \langle \F, P^k \rangle_t +   (\xi^k)^2  d\langle P^k \rangle_t \right)
\end{align*}
This expression is minimized by $\xi^{k}_t =\frac{d\langle P^k, \F\rangle_t}{d\langle P^k, P^k \rangle_t}$ as in \eqref{Eq_xik}.
\end{rem}

{
\section{Empirical analysis}
In this section we demonstrate the use of our polynomial framework for modeling and hedging long-term electricity forwards and analyzing their performance. Based on a time series of real observations of power forwards provided by Axpo Solutions AG, we estimate parameters of a model specification. Further we simulate forward curves and investigate the quality of risk-minimizing hedges over various time horizons.   

\subsection{The data}\label{subsec_1}

Electricity long-term contracts lack liquidity and are not available on exchange.\footnote{People usually refer to contracts with more than 2-3 years time to maturity/start of delivery as \emph{long-term} contracts.} In fact, long-term forwards with delivery periods are only offered by a small group of market participants over the counter (OTC), mostly by energy producing and trading companies. 

The data we use are provided by Axpo Solutions AG, and come originally from Totem Markit service, which surveys prices of various electricity contracts from each member firm and in term provides market consensus prices. More concretely, the data are German calender-year baseload (Cal) forwards that are quoted monthly from January 2010 to April 2018.\footnote{Note that German Cal Base forwards are the most liquidly traded contracts among all illiquid long-term contracts.} On each quotation date, we have at most 10 quoted contracts, i.e.\ first to tenth nearby Cal forwards. For each quoted contract on each quotation date, we have consensus price and the price spread between the highest quoted price and the lowest quoted price. A visualization of consensus prices is given in  \sref{Figure}{fig:ds}. 

\begin{figure}[ht]
\begin{subfigure}{0.49\textwidth}
\includegraphics[scale=0.24]{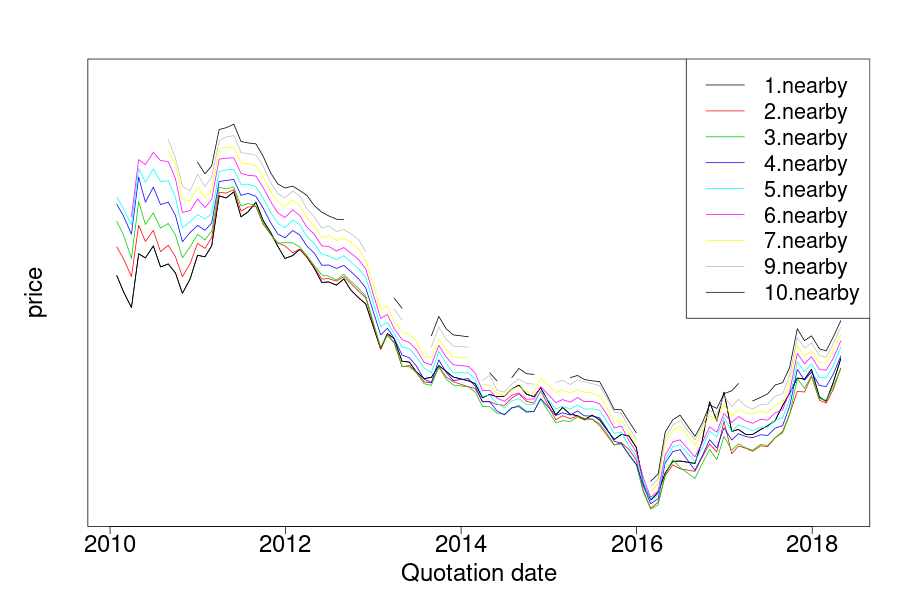}
\caption{rolling forwards} \label{fig:ds_a}
\end{subfigure}
\begin{subfigure}{0.49\textwidth}
\includegraphics[scale=0.24]{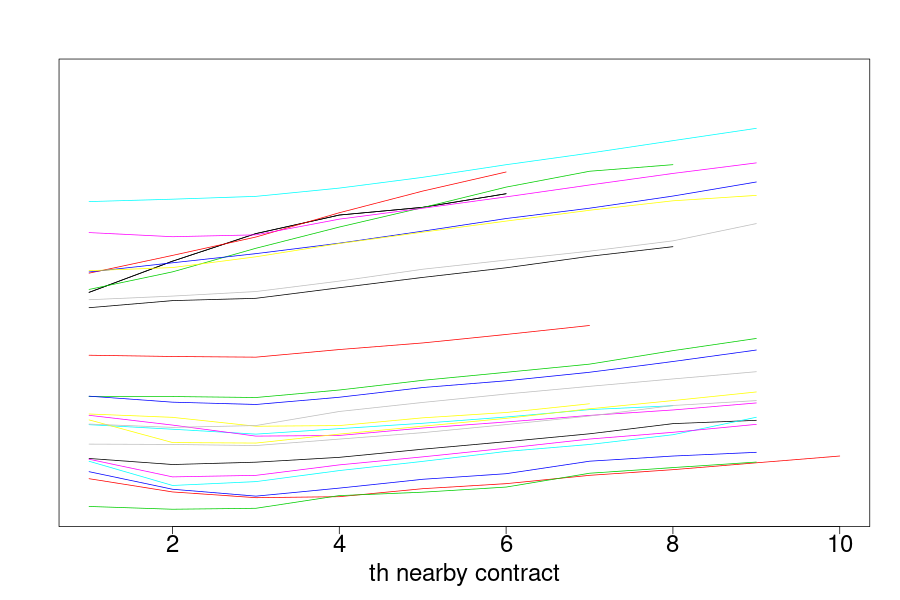}
\caption{a selection of forward curves} \label{fig:ds_b}
\end{subfigure}
\caption{~(color online). German Calendar-year Baseload forward from January 2010 to April 2018. Y-axes are removed for data protection. Figure (a) shows the dynamics of each nearby Cal forward contract with respect to quotation date. We see that not every contract is available on every quotation date. In Figure (b), 
each curve is the forward curve of a quotation date, i.e.\ each curve shows the prices of the first to at most tenth nearby Cal forwards of that date. For the sake of a clearer view, we take a selection of forward curves. These curves (of chosen quotation dates) are stacked and time-lagged into a day. We note two shapes of forward curves: a straight contango curve and a curve which is flat with slight backwardation at the front and contango at the back end of the curve.   
}\label{fig:ds}
\end{figure}

\subsection{Model estimation}\label{subsec_KF}

In order to capture the dynamics of the forward curves with our model, a non-linear filter is needed for model estimation, as the forward prices are quadratic in the Gaussian underlying factor process $X_t$. Recall that the fundamental assumption of Kalman filter is that the measurement space is linear and Gaussian in the state space. Thus, in order to work with a Kalman filter, we can either linearize the quadratic relationship between state and measurement. This leads to a so-called extended Kalman filter. Alternatively, we can augment the state to incorporate the linear and quadratic terms of $X_t$, so that the measurements become linear in the augmented state.

Inspired by the work of \cite{monfort2015quadratic}, we will use a time-dependent version of the latter approach to estimate a discrete version of \sref{Specification}{Sp_1}} based on the data from \sref{Section}{subsec_1}. The estimation will be under $\mathbb{P}$, which means that we also need to estimate the market price of risk parameters.

Note that we do not have direct access to the underlying state process $X_t$ through the available data. Indeed, at each quotation date $t_k$, we only see the prevailing price $F^j_k$ of the $j$-th nearby forward contract, with $j=1,\ldots,10$.\footnote{Actually, we see even less, since price data is often missing for longer maturities.} We view $F^j_k$ as a noisy observation of the model price. More precisely, we assume that
\[
F^j_k = F(t_k,T_j,T_j+1) + N^j_k \eta^j_k,
\]
where $F(t_k,T_j,T_j+1)$ is the model price computed using \sref{Proposition}{prop_fwd}, $\eta^j_k$ are iid standard Gaussian noise, modulated by some parameters $N^j_k>0$. The role of $N^j_k$ is to encode the trustworthiness of the price of the $j$-th nearby contract on quotation date $t_k$. A large value means that the price is considered noisy and uncertain, and a small value that the price is considered accurate. The $N^j_k$ are chosen based on the spreads $\delta^j_k$ between the highest and lowest quoted price for the $j$-th nearby contract on date $t_k$. Specifically, we use
\[
(N^j_k)^2 = \frac13 \times \delta^j_k + \frac13 \times \delta^j + \frac13 \times \delta,
\]
where $\delta^j$ denotes the time series average of the spreads $\delta^j_k$ for a fixed maturity $j$, and $\delta$ denotes the overall average of all the spreads $\delta^j_k$. The use of iid noise corresponds to assuming that our model captures all systematic effects. This is a standard assumption to reduce the complexity of the estimation.

\medskip
\paragraph{A quadratic Kalman filter for \sref{Specification}{Sp_1}}
We will now overload notation in the following manner: we write $X_k$ for the state $X_{t_k}$ at quotation date $t_k$, and similarly for other quantities that depend on time.

Since model prices at date $t_k$ are quadratic in the state $X_k$, we have the expression
\[
F^j_k = a^j_{k} + B^j_{k} X_{k} + X_k^\top C^j_{k} X_{k} + N^j_k \eta^j_k
\]
for some $a^j_{k}\in\R$, $B^j_{k}\in\R^{2}$ and $C^j_k\in\S^2$ that can be deduced from the pricing formula in \sref{Proposition}{prop_fwd}. In view of \eqref{Eq_H1}, and following \cite{monfort2015quadratic}, we observe that $F^j_k$ is affine in the augmented state vector
\[
\widetilde X_k = (Z_k,\ Y_k,\ Z_k^2,\ Y_kZ_k,\ Y_k^2)^\top.
\]
Specifically, the vector of prices, $F_k=(F^1_k,\ldots,F^{10}_k)^\top$ is given by
\[
F_k = a_k + \widetilde{B}_k \widetilde{X}_{k} + N_k\eta_{k},
\]
where $a_k:=(a^1_k,\ldots,a^{10}_k)^\top$ and $\widetilde{B}_k:= (\widetilde{B}^1_k,\ldots, \widetilde{B}^{10}_k)^\top$ can be computed as follows: for each maturity $j=1,\ldots,10$, we have
\[
\begin{pmatrix} a^j_{k} \\ \widetilde{B}^{j}_{k} \end{pmatrix} :=e^{(T_j-{t}){G}}\int^1_0 e^{u{G}}du \,\vec{p}_S,
\]
with $\vec p_S$ from (\ref{Eq_pS1}) and $G$ from (\ref{Eq_G1}). Moreover, we have defined $N_k :=\diag\left( N^1_k,\ldots, N^{10}_k\right)$ and $\eta_k:=(\eta^1_k,\ldots,\eta^{10}_k)^\top$. Next, the discretized (non-augmented) state dynamics is given by
\[
X_{k} = b + D X_{k-1}+K\varepsilon_{k}
\]
where $\varepsilon_{k}$ are independent bi-variate standard Gaussians and
\begin{align*}
 b &= \scalemath{0.95}{
 \begin{pmatrix}
     \gamma_Z\Delta t \\
     \gamma_Y \Delta t \\ 
    \end{pmatrix}},
  D = \scalemath{0.95}{\begin{pmatrix}
     1-(\kappa_Z - \lambda_Z)\Delta t & 0 \\
       \kappa_Y\Delta t & 1-(\kappa_Y- \lambda_Y)\Delta t \\
    \end{pmatrix}},
K = \scalemath{0.95}{\begin{pmatrix}
\sigma_Z \sqrt{\Delta t} & 0\\
\rho\sigma_Y \sqrt{\Delta t} & \sigma_Y\sqrt{(1 - \rho^2)\Delta t}   
\end{pmatrix}}.
\end{align*}
Here we use the market price of risk parameters $\Lambda = \diag(\lambda_Z, \lambda_Y)$ and $\gamma = (\gamma_Z, \gamma_Y)^\top$ from \sref{Section}{sec_lambda}. The discretized dynamics of the augmented state $\widetilde X_k$ is
\begin{align*}
\widetilde{X}_{k} = \widetilde{b}(X_{k-1}) + \widetilde{D} \widetilde{X}_{k-1} + \widetilde{K}(X_{k-1}) \varepsilon_{k},
\end{align*} 
where the involved quantities are conveniently expressed using the standard vector stacking operator $Vec()$, Kronecker product $\otimes$, selection matrix $H_d$, and duplication matrix $G_d$. The resulting expressions are:
\begin{align*}
\widetilde{b}(X_{k-1}) &= 
\begin{pmatrix} 
b \\ H_2 Vec(b b^\top + \Sigma )
\end{pmatrix}, \quad
\widetilde{D} = 
\begin{pmatrix}
D & 0  \\ 
H_2( b\otimes D+ D \otimes b) G_2 & H_2 ( D\otimes D )G_2 
\end{pmatrix},\\
\Gamma_{k-1}&= I_2\otimes (b+DX_{k-1}) + (b+DX_{k-1})\otimes I_2,\\
\widetilde{\Sigma} (X_{k-1}) &= 
    \begin{pmatrix}
    \Sigma &    \Sigma \, \Gamma_{k-1}^\top H_2^\top  \\
    H_2 \,\Gamma_{k-1} \,\Sigma  &  H_2 \Gamma_{k-1}\Sigma \Gamma_{k-1}^\top H_2^\top + H_2 (I_4+ \Lambda_2)(\Sigma\otimes \Sigma)H_2^\top 
    \end{pmatrix},
\end{align*}
where $\Sigma:=KK^\top$ and $I_d$ is the identity matrix of size $d$, and $\Lambda_m $ is the standard commutation matrix of size ${m^2\times m^2}$. We then let $\widetilde{K}(X_{k-1})$ be the Cholesky factor of $\widetilde{\Sigma}(X_{k-1})$, i.e., $\widetilde{K}(X_{k-1})\widetilde{K}(X_{k-1})^\top= \widetilde{\Sigma}(X_{k-1})$. We finally define $\mathcal{F}_{k-1}:= \sigma(F_{k-1},F_{k-2},..., F_{1} )$. The filtering algorithm is then described in \sref{Algorithm}{Algo_qKF}, where we use the notation
\begin{align*}
\widetilde{X}_{k|k-1} :&= \mathbb{E}[\widetilde{X}_{k}|\mathcal{F}_{k-1}], & \widetilde{V}_{k|k-1} &:=\mathbb{V}[\widetilde{X}_{k}|\mathcal{F}_{k-1}],\\
{F}^j_{k|k-1} :&= \mathbb{E}[F^j_{k}|\mathcal{F}_{k-1}], & M^j_{k|k-1} &:=\mathbb{V}[F^j_{k}|\mathcal{F}_{k-1}].
\end{align*}

\begin{algorithm}[H]
\caption{Quadratic Kalman filtering algorithm}
\begin{algorithmic}
\STATE {Anchoring}: 
\STATE \eqmakebox[cse][l]{$~~\widetilde{X}_{1|1}$} $=\widetilde{x}_0 = (x_0^\top, H_2 Vec(x_0 x_0^\top))^\top=(z_0, y_0, z^2_0, y_0z_0, y^2_0)^\top $,  
\STATE \eqmakebox[cse][l]{$~~\widetilde{V}_{1|1}$} $= \widetilde{\Sigma}({x}_{0})$.
\STATE {State prediction}:
\STATE \eqmakebox[cse][l]{$~~\widetilde{X}_{k|k-1}$} $= \widetilde{b}(X_{k-1|k-1}) + \widetilde{D} \widetilde{X}_{k-1|k-1} $,
\STATE \eqmakebox[cse][l]{$~~\widetilde{V}_{k|k-1}$} $= \widetilde{D}  \widetilde{V}_{k-1|k-1} \widetilde{D}^\top + \widetilde{\Sigma}(X_{k-1|k-1}) 
$.
\STATE {Measurement prediction}:
\STATE \eqmakebox[cse][l]{$~~F_{k|k-1}$} $= a_k+ \widetilde{B}_{k}\widetilde{X}_{k|k-1}$.
\STATE \eqmakebox[cse][l]{$~~M_{k|k-1}$} $= \widetilde{B}_{k} \widetilde{V}_{k|k-1}  \widetilde{B}^{\top}_{k} +N_k {N}_k^\top $.
\STATE \eqmakebox[cse][l]{$~~\mathcal{C}_{k}$} $= (F^{\,\text{real}}_{k} - F_{k|k-1})$ gives the prediction error.
\STATE {Update}: 
\STATE \eqmakebox[cse][l]{$~~\mathcal{K}_{k}$} $= \widetilde{V}_{k|k-1}\widetilde{B}^\top_{k} M_{k|k-1}^{-1} $ gives the gain matrix, 
\STATE \eqmakebox[cse][l]{$~~\widetilde{X}_{k|k}$} $=\widetilde{X}_{k|k-1} + \mathcal{K}_{k}\mathcal{C}_{k}  $,
\STATE \eqmakebox[cse][l]{$~~\widetilde{V}_{k|k}$} $= \widetilde{V}_{k|k-1} -\mathcal{K}_{k} M_{k|k-1}\mathcal{K}^\top_{k}= (\mathbbm{1}-\mathcal{K}_{k}\widetilde{B}_{k})\widetilde{V}_{k|k-1} $,
\STATE \eqmakebox[cse][l]{$~~F^j_{k|k}$} $= a_{k}+ \widetilde{B}_{k}\widetilde{X}_{k|k}$. 
\end{algorithmic}
\label{Algo_qKF}
\end{algorithm}

\medskip
\paragraph{Optimization with the quadratic Kalman filter for \sref{Specification}{Sp_1}}
For the model estimation with the quadratic filter, we use both the Least-Squares (LS) and the Maximum Likelihood (ML) criteria. We start with LS, as it is robust and converges fast. Once a stable result is obtained, we apply ML to obtain further improvement. Moreover, we impose $1\geq\kappa_Y \geq\kappa_Z \geq 0$ on the parameters, in line with the interpretation that $Y_t$ and $Z_t$ drive the short and the long end of the forward curve respectively and thus mean-revert at different speed. The filtered underlying process $X_t=(Z_t,Y_t)^\top$ is given in \sref{Figure}{fig:underlying}. The estimated parameters are shown in \sref{Table}{tab:param}.

\begin{figure}[H]
\centering
\includegraphics[scale=0.315]{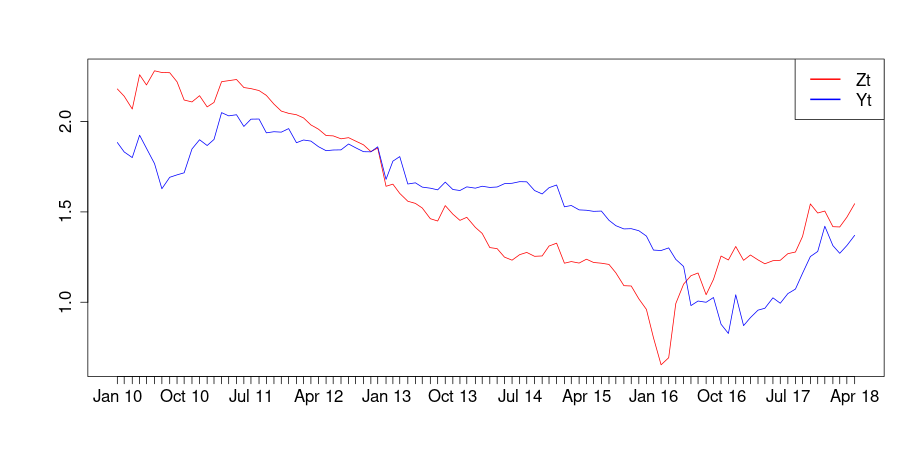}
\caption{~(color online). The filtered underlying dynamics $X_t=(Z_t,Y_t)^\top$ of \sref{Specification}{Sp_1}. 
}
\label{fig:underlying}
\end{figure}

\begin{table}[H]
\begin{center}
\begin{tabular}{m{2cm} | r  }
  \hline\hline
 $c$ & $0.239614$  \\ \hline
 $\alpha$ & $10.250035$ \\  \hline
 $\beta$ & $0.176807$ \\ \hline 
 $\kappa_Z$ & $0.010022$\\ \hline 
 $\kappa_Y$ & $0.400207$\\ \hline 
 $\sigma_Z$ & $0.406479$\\ \hline 
 $\sigma_Y$ & $0.889130$\\ \hline 
 $\rho$ & $0.112439$\\ \hline
 $\lambda_Z$ & $0.089990$\\ \hline
 $\lambda_Y$ & $0.111842$\\ \hline
 $\gamma_Z$ & $0.086791$\\ \hline
 $\gamma_Y$ & $0.127365$\\ \hline
 $z_0$ & $2.358048$\\ \hline
 $y_0$ & $2.007557$ \\ \hline
 \end{tabular} 
 \caption{Estimated parameters of \sref{Specification}{Sp_1}.}\label{tab:param}
\end{center}
\end{table}

In the implementation we use the R package DEoptim, which is an optimizer based on a differential evolution algorithm; see \cite{Storn1997}, \cite{DEoptim2006} for details of the algorithm and \url{https://cran.r-project.org/web/packages/
DEoptim/index.html}, \cite{DEoptim2011c}, \cite{DEoptim2016}, \cite{DEoptim2011a} \cite{DEoptim2011b} for use of the package.

\sref{Figure}{fig_me} gives a visualization of the model estimation using \sref{Specification}{Sp_1}. We quantify the goodness of fit in terms of relative errors, both cross-sectionally at each quotation date (\sref{Figure}{fig_boxplot2}), and across time for each nearby forward contract (\sref{Figure}{fig_boxplot1}). The overall relative error, i.e.\ the average relative error across all contracts and quotation dates, is as low as $0.661\%$, indicating a very good model fit.

\begin{figure}[ht]
\begin{subfigure}{0.49\textwidth}
\includegraphics[scale=0.24]{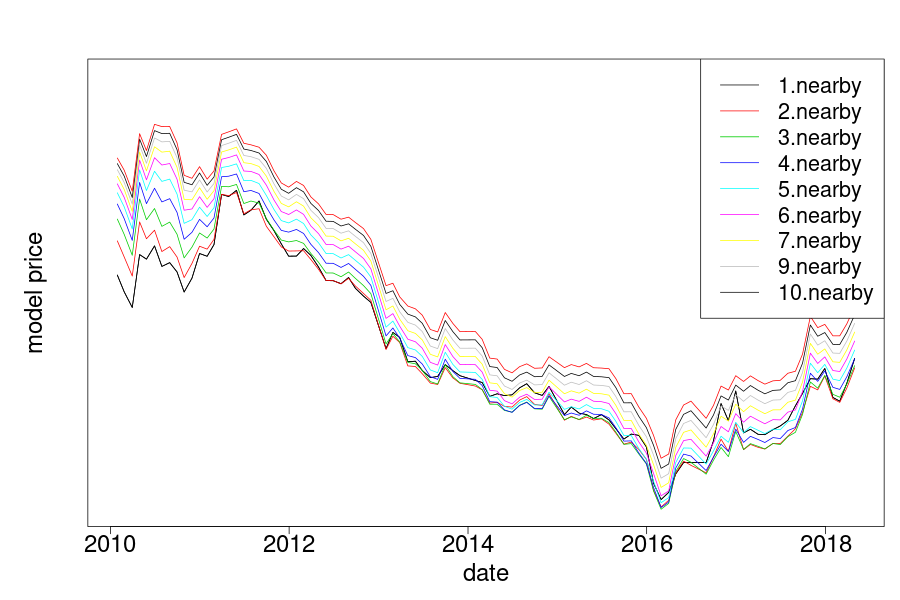}
\caption{model estimated rolling forwards} 
\end{subfigure}
\begin{subfigure}{0.49\textwidth}
\includegraphics[scale=0.24]{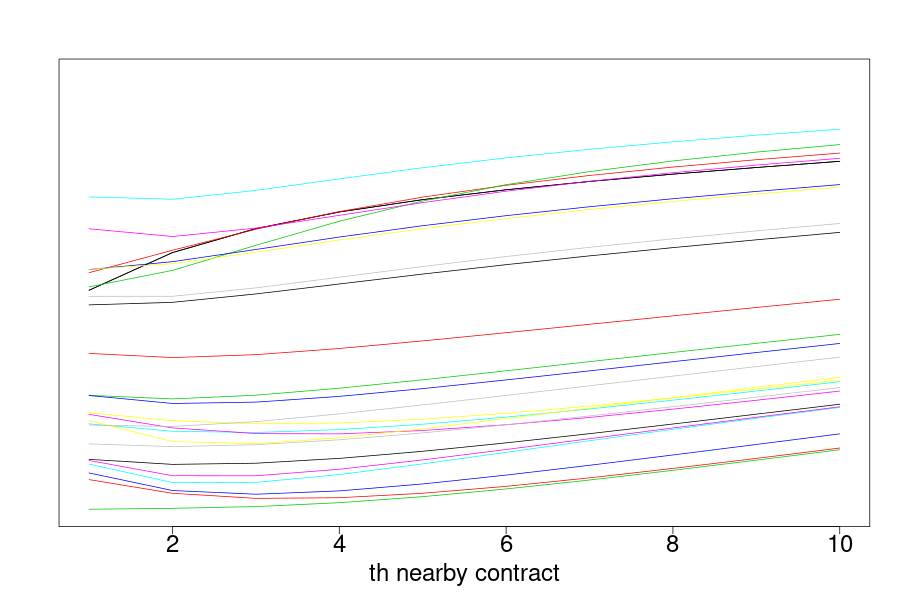}
\caption{a selection of model estimated forward curves} 
\end{subfigure}
\caption{~(color online). Forward curves from \sref{Specification}{Sp_1} using estimated parameters: in (a) each nearby forward is shown as a time series; in (b) each curve is a forward curve at a particular quotation date (same date selection as in \sref{Figure}{fig:ds_b}). Y-axes are removed for data protection. Comparing these figures with the real observations (\sref{Figure}{fig:ds}), we find that the model captures the shapes and dynamics of the time series observation of electricity forward curves well.}
\label{fig_me}
\end{figure}

In \sref{Figure}{fig_boxplot2} we notice a single spike of the time series of averaged errors reaching almost $2\%$ (on a quotation date in February 2016). This is due to a single dramatic price drop of a forward curve on that date that is moderately captured by our model as it is continuous and gives smooth prices.

Looking at the estimation of the time series of each nearby forward (\sref{Figure}{fig_boxplot1}), we find that the front end fit (i.e. the first nearby to the sixth nearby forward contract) works very well while the fit deteriorates for longer maturities. This occurs by construction, as the prices of contracts with very long time-to-maturity are less reliable than those on the front end of the forward curve. In the filter this is captured by the data variance $N^j_t$, which is influenced by the time series of price spread of each forward; in general $N^j_t$ tends to be higher for longer time-to-maturity (i.e. larger $j$).

\begin{figure}[ht!]
\begin{subfigure}{0.49\textwidth}
\includegraphics[scale=0.24]{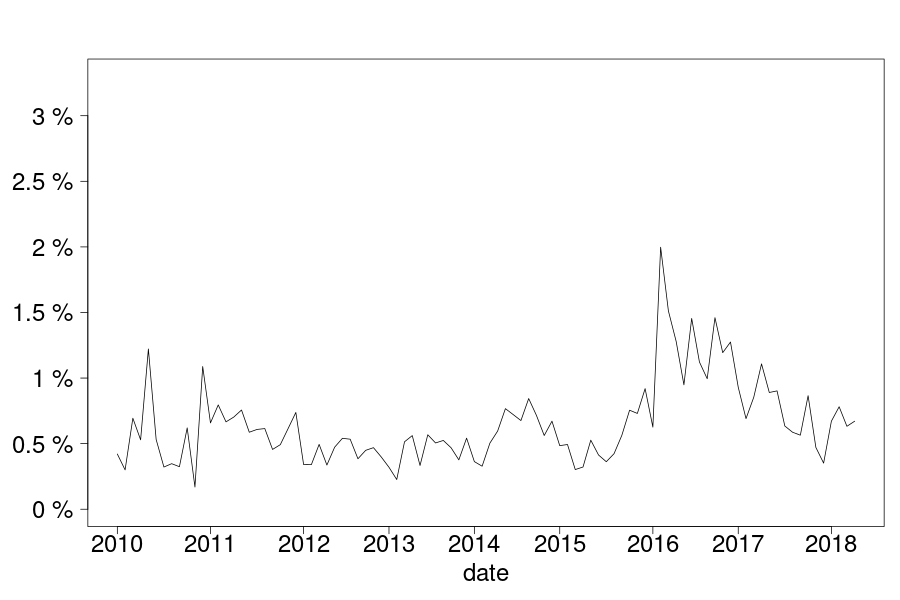}
\caption{averaged relative errors with respect to \\quotation date} \label{fig_boxplot2}
\end{subfigure}
\begin{subfigure}{0.49\textwidth}
\includegraphics[scale=0.24]{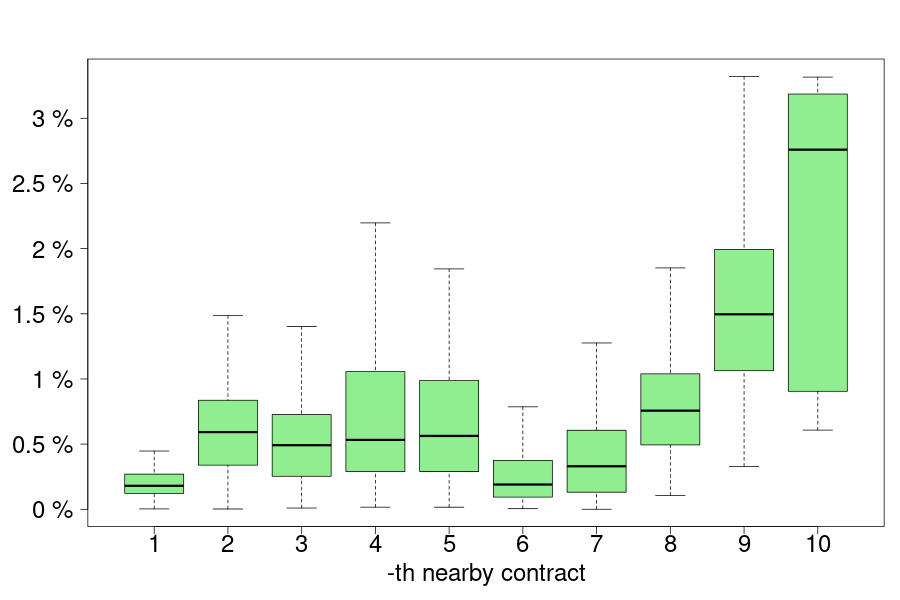}
\caption{first to third quantile of time series of relative errors with respect to rolling contract} \label{fig_boxplot1}
\end{subfigure}
\caption{~(color online). Relative errors of model estimation. The overall relative error (averaged over all contracts and all quotation dates) is $0.661\%$. In (a) the averaged relative error of forward curve on each quotation date is shown. The spike in February 2016 is caused by a large downward drop of the observed forward curve, leading the model to deviate $2\%$ on average on that date. In (b) the distribution of relative pricing errors for each nearby contract over time is given in boxplot: each whisker gives the range from mininum value to maximum value of the time series of relative errors for that contract (outliers are removed). Each green box marks the 25th to 75th percentile of the time series. The thick black line marks the median relative error. In addition to (b), the time averaged relative errors and standard deviations for each contract are given in the table below. We see that the first to sixth nearby contracts are well estimated by the model, while the seventh to tenth nearby contracts have much larger estimation errors. This occurs by construction. The real data on the back end of the forward curve are very rare and thus have a huge price uncertainty; in particular the tenth nearby contract was only available on four quotation dates on over nine years of monthly quotation data. The uncertainty of real data is captured by the parameters $N^j_t$ for each $j$-th nearby rolling contract in the quadratic filter.}
\label{fig_boxplot}
\begin{center}
\begin{tabular}{c|c| c| c|c|c } 
 nearby contract &1& 2& 3 &4 &5\\ \hline \hline 
 av. rel. error& 0.2162\% &  0.6211\% &  0.5666\% &  0.7362\% &  0.6990\%
\\ \hline
std(rel. error) &  0.1741\% &  0.4036\% &  0.4401\% &  0.6223\% &  0.5560\% \\ \hline \hline
  nearby contract &6 &7& 8 & 9 & 10\\ \hline
av. rel error &  0.3355\% &  0.4509\% &  0.8530\% &  1.6583\% &  2.1549\% \\ \hline
std(rel. error)  & 0.4397\% & 0.4330\% & 0.5986\% & 0.9784\% & 1.2975\%  \\ \hline
  \end{tabular} 
\end{center}
\end{figure}

We also performed model estimation under $\mathbb{Q}$. This is equivalent to assuming $\mathbb{P}=\mathbb{Q}$, meaning that the market price of risk is zero ($\lambda(X_t)=0$). This produces different parameters than those in \sref{Table}{tab:param}, but the fit remains remarkably good. 


\subsection{Simulation and hedging analysis}

In the following, we simulate forward surfaces, run locally risk-minimizing hedging strategies on those, and analyze their performance with respect to different hedging horizons. 
 
\medskip
\paragraph{Simulation of forward surfaces}
With a given set of parameters, we generate samples of entire forward surfaces over a fixed time horizon $\T$. This can be done efficiently by first simulating the $\mathbb{P}$-dynamics of the underlying process $X_t= (Y_t, Z_t)^\top$ until year $\T$ using a simple Euler scheme (with, say, $N$ discretization steps). We can then compute the forward price for the $1$-st through $L$-th nearby contract at each point $t\le\T$ on the time grid by applying the pricing formula, \sref{Proposition}{prop_fwd}. The complexity of simulating $M$ evolutions of forward curves is of the order $\O((M \times N)^L)$. A brief pseudo code is given in \sref{Algorithm}{Algo_fsurface}. 

\begin{algorithm}[h]
\caption{Simulate forward surfaces under $\mathbb{P}$ (with market price of risk)}
\begin{algorithmic}
\REQUIRE $ \varepsilon^Y_j,\varepsilon^Z_j\stackrel{iid}{\sim} \mathcal{N}(0,1)$, $j = 1,..., N$, $\T$, $M$, $N$, $L$ and all model parameters (see e.g. Table \ref{tab:param}).  
\ENSURE $M$ simulated forward surfaces over $\T$ years. 
\STATE $\Delta t = T/N$
\STATE $Y_{0} = y_0$ 
\STATE $Z_{0} = z_0$
\STATE $H(X_{0}) = (1,z_0, y_0, z_0^2, y_0z_0, y_0^2 )^\top$
\FOR{$l =1,...,L$} 
	\STATE $F_{0}^l = H(X_{0}) \, e^{lG}\, \vec w_{0,1} $ 
\ENDFOR
\FORALL{ $M$ simulations }
	\FOR{$j = 1,..., N$}
		\STATE	$Z_{j} = \gamma_Z\Delta t +\left(1 - (\kappa_Z-\lambda_Z)\Delta t\right) Z_{j-1} + \sigma_Z\sqrt{\Delta t }\,\varepsilon^Z_{j} $
		\STATE	$Y_{j} = \gamma_Y\Delta t + \kappa_Y\Delta t Z_{j-1} + (1-(\kappa_Y-\lambda_Y)\Delta t) Y_{j-1}+ \sigma_Y \sqrt{\Delta t } (\rho  \varepsilon^Z_{j} +\sqrt{1-\rho^2} \varepsilon^Y_{j} )$
		
		\STATE $H(X_{j}) = (1,Z_{j}, Y_{j}, Z_{j}^2, Y_{j} Z_{j}, Y_{j}^2 )^\top$ 
		\FOR{$l =1,...,L$}
			\STATE  $F_{j}^l = H(X_{j})\, e^{(l-(j\Delta t \mod 1 ))G} \, \vec w_{0,1} $
		\ENDFOR	
	\ENDFOR
\ENDFOR
\end{algorithmic}
\label{Algo_fsurface}
\end{algorithm}

\medskip
\paragraph{Simulation study of hedging performance}
We aim to evaluate hedging performance by comparing the unhedged exposures with exposures when we use the locally risk-minimizing rolling hedges from \sref{Section}{sec_hedging} on different hedging horizons. For this, we consider different claims $F(t, T):=F(t, T, T+1, X_t)$ with $T = 2,\ldots,10$ years. Next, we simulate $M=5000$ forward curve evolutions using the estimated parameters from \sref{Table}{tab:param}. For the Euler discretization we use 120 time points per year. For the hedging we use a monthly rebalancing frequency. Finally, we compare the percentage exposure if left unhedged, i.e.
$$\dfrac{F(T,T) - F(0,T)}{F(0,T)},$$
with the percentage exposure if hedged, i.e.
$$\dfrac{F(T,T) - F(0,T) - \int_0^T {\xi^{rm}_t}^\top dP_s }{F(0,T)},$$ with $\xi^{rm}_s$ from \eqref{Eq_xi}--\eqref{Eq_xik} and $P_s$ from \eqref{Eq_Pt}--\eqref{Eq_Pt1}. 
A visual comparison of those exposures (hedged versus unhedged) with respect to different hedging horizons is given in \sref{Figure}{fig_simhedgeP}. We see that the distribution of the exposure widens with increasing hedging horizon, and that the sample standard deviation and skewness go up; see the table below \sref{Figure}{fig_simhedgeP}. The exposure is significantly higher if left unhedged. Moreover, in all cases, the locally risk-minimizing rolling hedge significantly reduces, but does not eliminate, the variance and skew of long-term exposures.

\begin{figure}[ht]
\begin{center}
\includegraphics[scale=0.6]{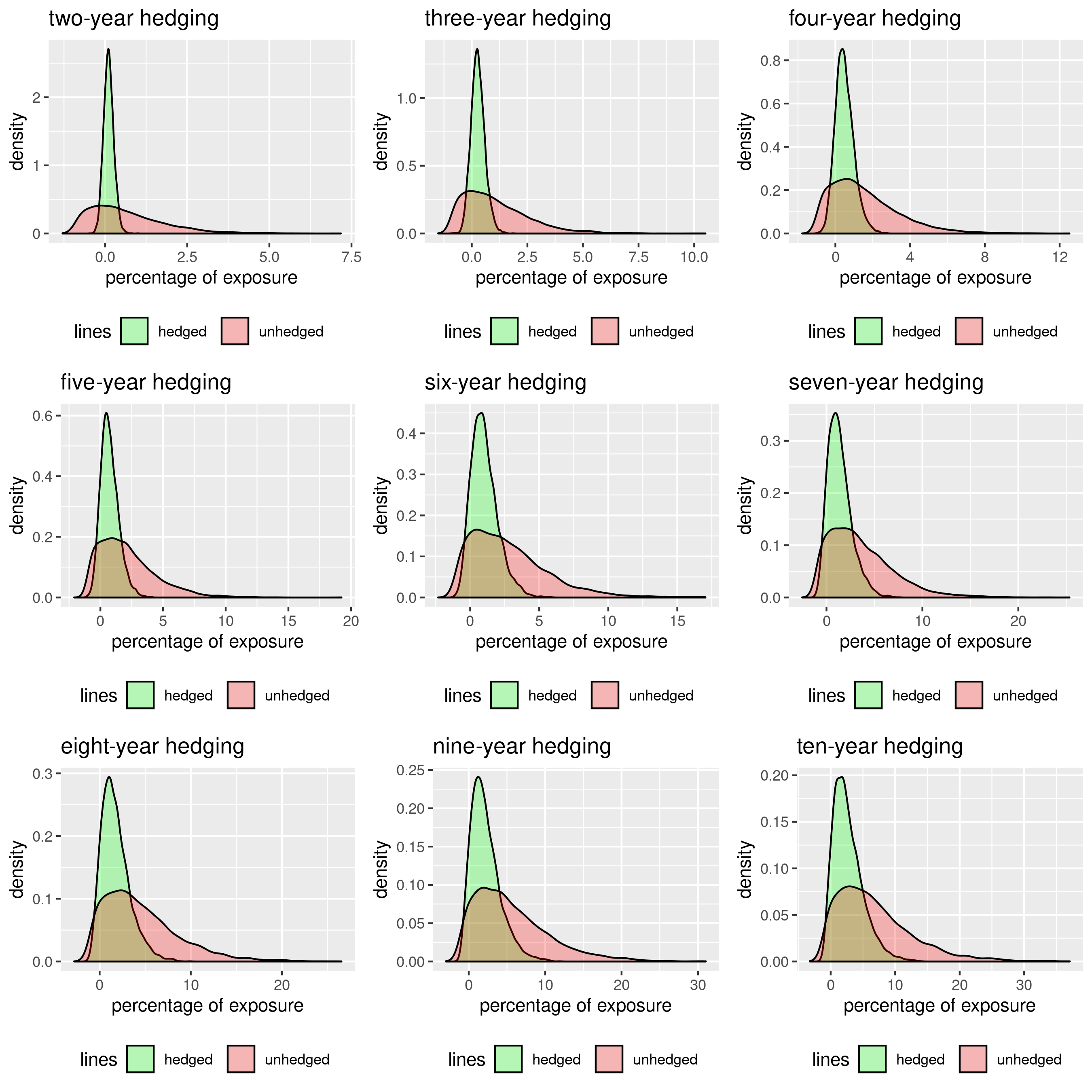}
\caption{~(color online). Density of hedged exposure (green) versus that of unhedged exposure (red) with respect to different hedging horizons. Forwards are simulated using the $\mathbb{P}$-dynamics and market price of risk. In each figure, a different forward is simulated such that the time to maturity corresponds the hedging horizon: i.e. in top left figure, we simulate a forward that matures and starts delivery in two years and compute the exposure at maturity; we then compute a risk-minimizing hedge (with two years hedging horizon), the hedged exposure, and obtain the comparison. Standard deviations and skewnesses are reported in the table below.}
 \label{fig_simhedgeP}
 \begin{tabular}{c| c|c l|c|c l}
  \multirow{2}{*}{hedging horizon} & \multicolumn{3}{c|}{ hedged} & \multicolumn{3}{c}{unhedged}\\ \cline{2-7}
 & std & skew && std & skew & \\ \hline \hline
 2 years & 0.1532 & 0.2728 && 1.1278 & 1.1724 \\
 3 years & 0.3099 & 0.3658 && 1.4700 & 1.2107 \\
 4 years & 0.4959 & 0.5477 && 1.8143 & 1.1738 \\
 5 years & 0.7125 & 0.6992 && 2.2762 & 1.2201 \\
 6 years & 0.9583 & 0.8474 && 2.8011 & 1.2439 \\
 7 years & 1.2266 & 0.9061 && 3.3729 & 1.2361 \\
 8 years & 1.5406 & 1.0017 && 4.0898 & 1.1926 \\
 9 years & 1.8991 & 1.0625 && 4.8472 & 1.1660 \\
 10 years& 2.2982 & 1.0777 && 5.7729 & 1.2224\\ \hline
 \end{tabular}
 \end{center}
\end{figure}


\clearpage

\appendix

\section{Explicit computation of $\boldsymbol{\int^t_0 e^{Gs}ds}$}\label{Appen_eGs}
The $G$-matrices arising in both specifications have a zero first column, and are therefore not invertible. This is in general the case when $1$ is part of the basis $H(x)$, as $\G 1 = 0$. Moreover, if we remove the first row and column of $G$, the submatrix $G^{\prime}$ is invertible and upper-triangular. In the following we show a straightforward way to compute $\int^t_0 e^{Gs}ds$ for such $G$, which helps to reduce the computational effort of evaluating the pricing formula.

\begin{prop}\label{prop_intG} Let A be an upper triangular matrix of the form
\[
A = \begin{pmatrix} 0 & b^\top \\ \vec 0 & C \end{pmatrix}
\]
for some vector $b$ and upper triangular invertible matrix $C$. Then
\[
e^{At} = \begin{pmatrix} 1 & b^\top C^{-1}(e^{Ct}-I) \\ \vec 0 & e^{Ct} \end{pmatrix}
\quad\text{and}\quad
\int_0^t e^{As}ds = \begin{pmatrix} t & b^\top (C^{-1})^2(e^{Ct}-I)-t b^\top C^{-1} \\ \vec 0 & C^{-1}(e^{Ct}-I) \end{pmatrix}
\]
\end{prop}

\begin{proof}
Let $F(t)$ denote the claimed expression for $e^{At}$. One easily checks that $F'(t)=AF(t)$ and that $F(0)$ is the identity. This implies that $F(t)=e^{At}$. The expression for $\int_0^t e^{As}ds$ is easily obtained by integrating each block of $F(t)$.
\end{proof}

\section{Specifications of $\boldsymbol{\Sigma (X_t)}$}\label{append_Sigma}
\paragraph{Instantaneous covariations and correlations in \sref{Specification}{Sp_1}} Equations (\ref{Eq_covwodel}), (\ref{Eq_covwdel}), (\ref{Eq_corwodel}), (\ref{Eq_corwdel}) hold with $H$ from (\ref{Eq_H1}), $\vec p_S$ from (\ref{Eq_pS1}) and $\Sigma (X_t)$ as below:
\begin{align*}
\Sigma_{X_t} = \left(
    \scalemath{0.65}{
    \begin{array}{cccccc}
 0 & 0 & 0 & 0 & 0 & 0\\ \\
 0 & \sigma^2_Z & \rho\sigma_Y\sigma_Z & 2\sigma^2_Z \, Z_t & \sigma^2_Z \, Y_t + \rho\sigma_Y\sigma_Z \, Zt & 2\rho\sigma_Y\sigma_Z \, Y_t\\ \\
 0 &\rho\sigma_Y\sigma_Z  & \sigma^2_Y & 2\rho\sigma_Y\sigma_Z \, Z_t & \sigma^2_Y \, Z_t + \rho\sigma_Y\sigma_Z \, Y_t & 2\sigma_Y^2\, Y_t\\ \\
 0 & 2\sigma^2_Z \, Z_t & 2\rho\sigma_Y\sigma_Z \, Z_t & 4\sigma^2_Z \, Z^2_t & 2\sigma^2_Z \, Y_t Z_t + 2\rho\sigma_Y\sigma_Z \, Z^2_t & 4\rho\sigma_Y \sigma_Z \, Y_t Z_t   \\ \\
 0 & \sigma^2_Z \, Y_t + \rho\sigma_Y\sigma_Z \, Zt  & \sigma^2_Y \, Z_t + \rho\sigma_Y\sigma_Z\, Y_t & 2\sigma^2_Z \, Y_tZ_t + 2\rho\sigma_Y \sigma_Z \, Z^2_t & ~\sigma^2_Z Y^2_t+\sigma^2_Y Z^2_t +2\rho\sigma_Y \sigma_Z \, Y_tZ_t   & 2\rho\sigma_Y \sigma_Z \, Y^2_t + 2\sigma^2_Y \, Y_t Z_t\\ \\
 0 & 2\rho\sigma_Y\sigma_Z\, Y_t & 2\sigma_Y^2\, Y_t & 4\rho\sigma_Y \sigma_Z \, Y_t Z_t  & 2\rho\sigma_Y \sigma_Z \, Y^2_t + 2\sigma^2_Y \, Y_t Z_t & 4\sigma^2_Y \, Y^2_t\\
    \end{array}
    }
  \right)
\end{align*}


 \paragraph{Instantaneous covariations and correlations in \sref{Specification}{Sp_2}} Equations (\ref{Eq_covwodel}), (\ref{Eq_covwdel}), (\ref{Eq_corwodel}), (\ref{Eq_corwdel}) hold with $H$ from (\ref{Eq_H2}), $\vec p_S$ from (\ref{Eq_pS2}) and $\Sigma (X_t)$ as below: 

\begin{align*}
\Sigma_{X_t} =\left(
    \scalemath{0.55}{
    \begin{array}{ccccccc}
 0 & 0 & 0 & 0 & 0 & 0 & 0\\ \\
 0 & \sigma^2_Z & \sigma_Y\sigma_Z\, R_t & 0 &  2\sigma^2_Z \, Z_t & \sigma^2_Z \, Y_t + \sigma_Y\sigma_Z\, R_t  Zt & 2\sigma_Y\sigma_Z\, R_tY_t\\ \\
 0 & \sigma_Y\sigma_Z\, R_t  & \sigma^2_Y & 0 & 2\sigma_Y\sigma_Z \, R_t Z_t & \sigma^2_Y \, Z_t + \sigma_Y\sigma_Z \, R_t Y_t & 2\sigma_Y^2\, Y_t\\ \\
  0 & 0 & 0 & \sigma^2_R(1 - R^2_t) & 0 & 0 & 0\\ \\
 0 & 2\sigma^2_Z \, Z_t & 2\sigma_Y\sigma_Z \, R_t Z_t & 0 & 4\sigma^2_Z \, Z^2_t & 2\sigma^2_Z \, Y_t Z_t + 2\sigma_Y\sigma_Z \, R_t Z^2_t & 4\sigma_Y \sigma_Z \, R_t Y_t Z_t \\  \\
 0 & \sigma^2_Z \, Y_t + \sigma_Y\sigma_Z \, R_t Z_t  & \sigma^2_Y \, Z_t + \sigma_Y\sigma_Z\, R_t Y_t & 0 &  2\sigma^2_Z \, Y_t Z_t + 2\sigma_Y \sigma_Z \, R_t Z^2_t & ~\sigma^2_Z Y^2_t+\sigma^2_Y Z^2_t +2\sigma_Y \sigma_Z \, R_t Y_t Z_t   & 2\sigma_Y \sigma_Z \, R_t Y^2_t + 2\sigma^2_Y \, Y_t Z_t\\ \\
 0 & 2\sigma_Y\sigma_Z\, R_t Y_t & 2\sigma_Y^2\, Y_t & 0 & 4\sigma_Y \sigma_Z \, R_t Y_t Z_t  & 2\sigma_Y \sigma_Z \, R_t Y^2_t + 2\sigma^2_Y \, Y_t Z_t & 4\sigma^2_Y \, Y^2_t
    \end{array}
    }
  \right)
\end{align*}

\section{Correlation of forwards implied by the data}\label{append_corr}

\begin{figure}[th]
\begin{center}
\includegraphics[scale=0.4]{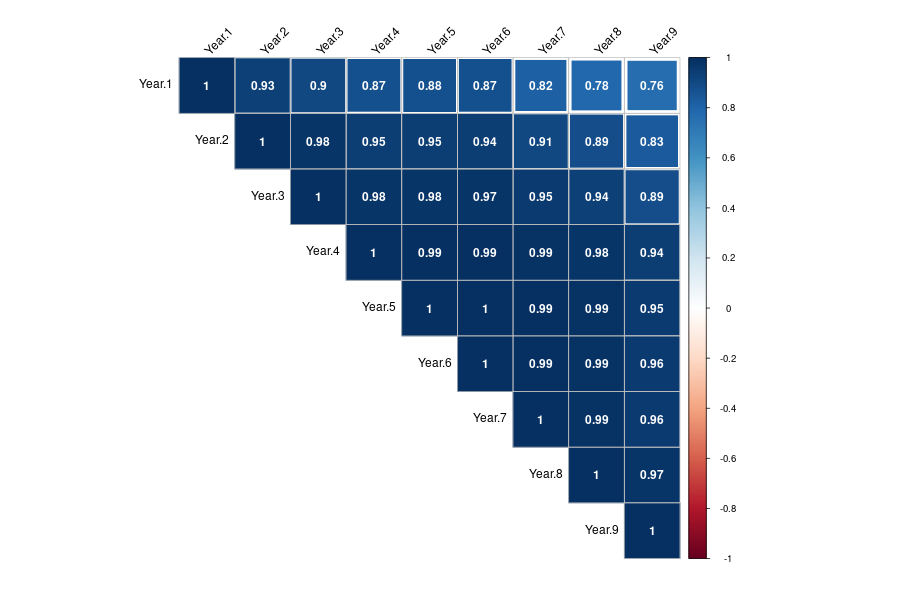}\caption{ ~(color online). Correlation between different nearby Calender year contracts implied by the data. }\label{Fig_corr}
\end{center}
\end{figure}

\clearpage

\bibliographystyle{plainnat}
\bibliography{Bibliographie}

\end{document}